\documentclass{article}

\usepackage{arxiv}

\usepackage[utf8]{inputenc} 
\usepackage[T1]{fontenc}   
\usepackage{hyperref}      
\usepackage{url}          
\usepackage{booktabs}      
\usepackage{amsfonts}      
\usepackage{nicefrac}      
\usepackage{microtype}     
\usepackage{lipsum}		
\usepackage{graphicx}
\usepackage{doi}
\usepackage{subcaption}
\usepackage{color}
\usepackage{graphics} 
\usepackage{epsfig}
\usepackage{amsmath} 
\usepackage{amssymb} 
\usepackage{breqn}
\usepackage{amsthm}

\usepackage[backend=biber,sorting=none]{biblatex}
\newtheorem{theorem}{Theorem}
\newtheorem{definition}{Definition}
\newtheorem{proposition}{Proposition}

\newtheorem{example}{Example}

\title{Rational Neural Network Controllers}

\author{Matthew Newton \\
	Department of Engineering Science\\
	University of Oxford\\
	Oxford, OX1 3PJ, UK \\
	\texttt{matthew.newton@eng.ox.ac.uk} \\
	\And
	Antonis Papachristodoulou\\
	Department of Engineering Science\\
	University of Oxford\\
	Oxford, OX1 3PJ, UK \\
	\texttt{antonis@eng.ox.ac.uk} \\
}

\date{}

\renewcommand{\headeright}{}
\renewcommand{\undertitle}{}
\renewcommand{\shorttitle}{Rational Neural Network Controllers}

\begin{document}
\maketitle

\begin{abstract} 
Neural networks have shown great success in many machine learning related tasks, due to their ability to act as general function approximators. Recent work has demonstrated the effectiveness of neural networks in control systems (known as neural feedback loops), most notably by using a neural network as a controller. However, one of the big challenges of this approach is that neural networks have been shown to be sensitive to adversarial attacks. This means that, unless they are designed properly, they are not an ideal candidate for controllers due to issues with robustness and uncertainty, which are pivotal aspects of control systems. There has been initial work on robustness to both analyse and design dynamical systems with neural network controllers. However, one prominent issue with these methods is that they use existing neural network 
architectures tailored for traditional machine learning tasks. These structures may not be appropriate for neural network controllers and it is important to consider alternative architectures. This paper considers rational neural networks and presents novel rational activation functions, which can be used effectively in robustness problems for neural feedback loops. Rational activation functions are replaced by a general rational neural network structure, which is convex in the neural network's parameters. A method is proposed to recover a stabilising controller from a Sum of Squares feasibility test. This approach is then applied to a refined rational neural network which is more compatible with Sum of Squares programming. Numerical examples show that this method can successfully recover stabilising rational neural network controllers for neural feedback loops with non-linear plants with noise and parametric uncertainty.
\end{abstract}

\section{Introduction} \label{sec:rationalNNintro}
Neural networks (NNs) have shown to be highly effective in numerous machine learning tasks. Examples of these include but are not limited to: image recognition, weather prediction, natural language processing, autonomous vehicle technology, medical imaging and social media algorithms \cite{czhang,cbis,bboc,tbro,skuu,mraz,tbal}. There have been numerous advancements that have contributed to their success such as the development of modern NN architectures \cite{akriz,resnet}, the increase in computational power available \cite{jsan,njou} and the availability of big data. 

More recently, there has been an increased interest in using NNs in control systems. One reason for this is the emergence of the parallel field of reinforcement learning. By harnessing the power of NNs, deep reinforcement learning has been used to create a decision-making agent to greatly outperform humans in many complex tasks. Such examples include the board game Go \cite{alphago} and the video game Dota 2 \cite{dota2}. Despite their success, there are many significant issues with theses methods. These learnt policies can perform poorly when the learnt environment is different from the real environment \cite{vbeh,aman,jmor,agle}. Additionally, bounds to quantify their safety do not sufficiently describe the performance of the algorithm and can be overly conservative \cite{rsut}. However, with new advancements in robust control and the success of NNs in reinforcement learning, there is a strong motivation for work at their intersection.

We refer to control systems that contain NNs as controllers as neural feedback loops (NFLs). Most research completed in this area has addressed the robustness analysis of NFLs, where the NN's parameters are given and the task is to quantify the system's safety or robustness properties. However, the more challenging task is to obtain the NN controller's parameters, whilst enforcing robustness guarantees. NNs are useful since they can be used as general function approximators \cite{khor,mtel,mles,fefa}, but contain a large number of parameters. This means that optimising over all of the parameters is often computationally expensive. 

One method to design the NN controller is by learning an expert control law using input-output data and then checking the robustness guarantees using analysis methods, \cite{hyin1,ppau2,mnew5}. However, this may lead to poor robustness certificates because no relevant objective is being optimised while training the NN. It is also possible to use reinforcement learning methods. This often involves training the controller by simulating the system trajectories and then updating the controller's parameters subject to maximising a reward function. However, there are significant challenges with this process; it is very computationally intensive, requiring a large amount of hyperparameter tuning and sometimes leads to undesirable behaviour \cite{mever2}. The parameterised NN controller can then be analysed to obtain robustness certificates in a defined operating region, however these guarantees can be poor. It can be difficult for these approaches to outperform traditional control laws. Furthermore, adding additional non-linearities into the model through the NN's activation functions can increase the complexity of the closed-loop system.

Despite these drawbacks, recent methods have focused on addressing these issues by designing NN controllers whilst ensuring robustness guarantees in the process. Methods that focus on developing reinforcement learning algorithms such as \cite{mever4,pdon} are able to create an NN policy which can be combined with robust control guarantees. However, these approaches still suffer from other reinforcement learning issues such as requiring significant computational time and hyperparameter tuning. These drawbacks can be mitigated by instead trying to obtain an NN controller by learning from an optimal model predictive control law. This allows the NN controller to be trained offline and when implemented it can be significantly less computationally expensive than the full model predictive control law. To achieve this an SDP framework that incorporates integral quadratic constraints is presented in \cite{hyin2}. An iterative algorithm is used to alternate between optimising the NN parameters to fit the control law and maximising the region of attraction. However, these approaches require a known model predictive control law to optimise over the controller's input-output data, which may not be easy to obtain. This approach relies on a loop transformation and Schur complement to ensure the optimisation problem is convex. Similar approaches have been used for different NN architectures such as recurrent NNs \cite{fagu} and recurrent equilibrium network controllers \cite{njun}. These methods also include a projected policy gradient algorithm and reinforcement learning to synthesise the controller, instead of requiring an expert control law.

\subsection{Our Contribution}
One shortfall of many recent approaches such as \cite{hyin2,fagu,njun} is that large convex approximations must be made, as the original problem is non-convex. For example, the ReLU and tanh activation functions must be sector bounded in the region $[0,1]$ with no constraints obtained from pre-processing bounds such as Interval Bound Propagation \cite{sgow}. These sector constraints are shown in Figure \ref{fig:sectorconstraints}. This can lead to the acquired robustness certificates being conservative. Another issue with these methods is that they rely on iterative approaches, where the algorithm alternates between optimising the performance of the controller and the stability guarantees. This can make the problem computationally expensive and intractable for large scale systems. They also require that the plant model is linear subject to sector bounded non-linearities. 

\begin{figure}[h]
    \centering
    \begin{subfigure}[b]{0.49\textwidth}
        \centering
        \includegraphics[width=\textwidth]{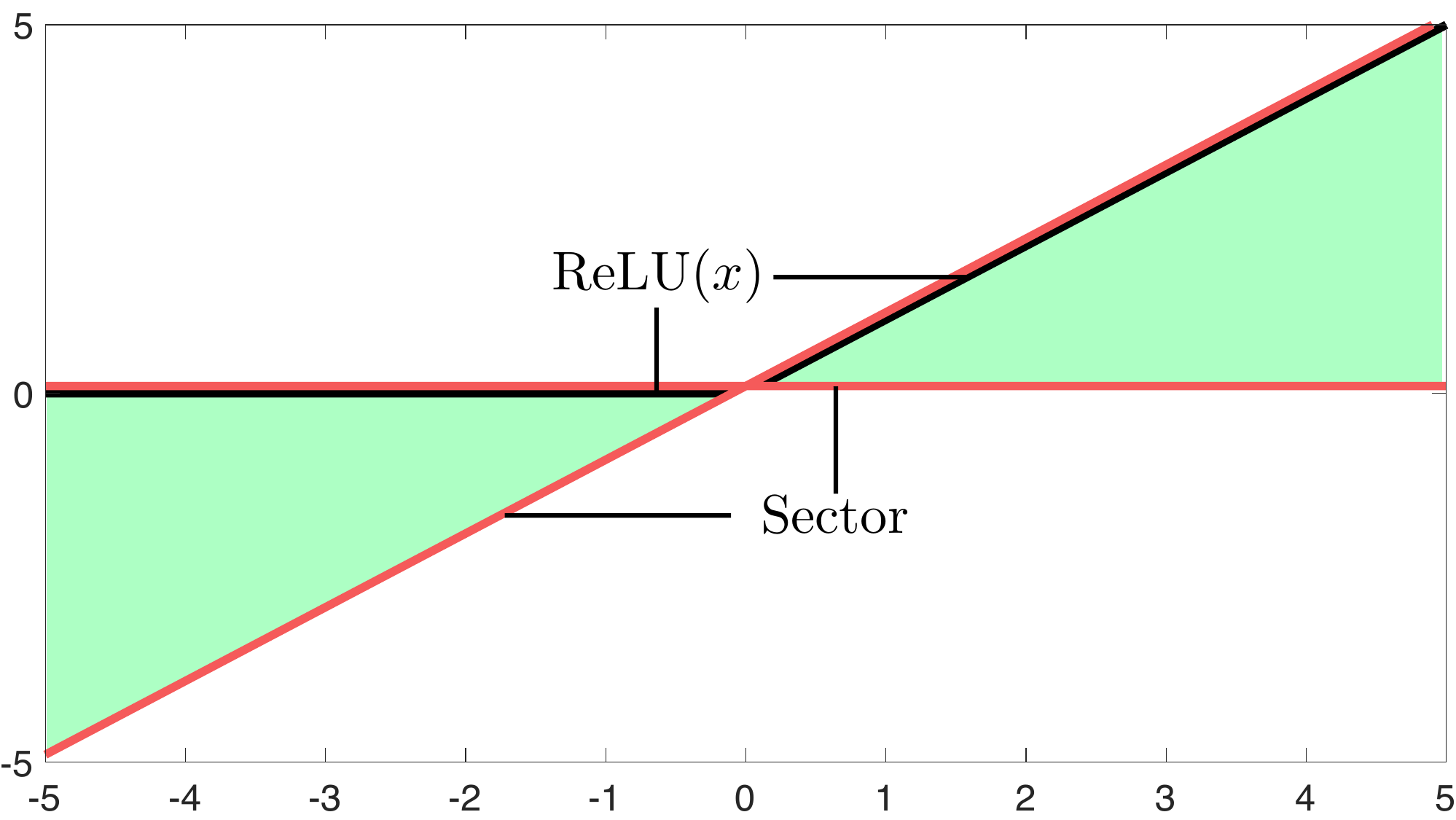}
        \caption{ReLU Activation Function}
        \label{fig:sectorrelu}
    \end{subfigure}
    \hfill
    \begin{subfigure}[b]{0.49\textwidth}  
        \centering 
        \includegraphics[width=\textwidth]{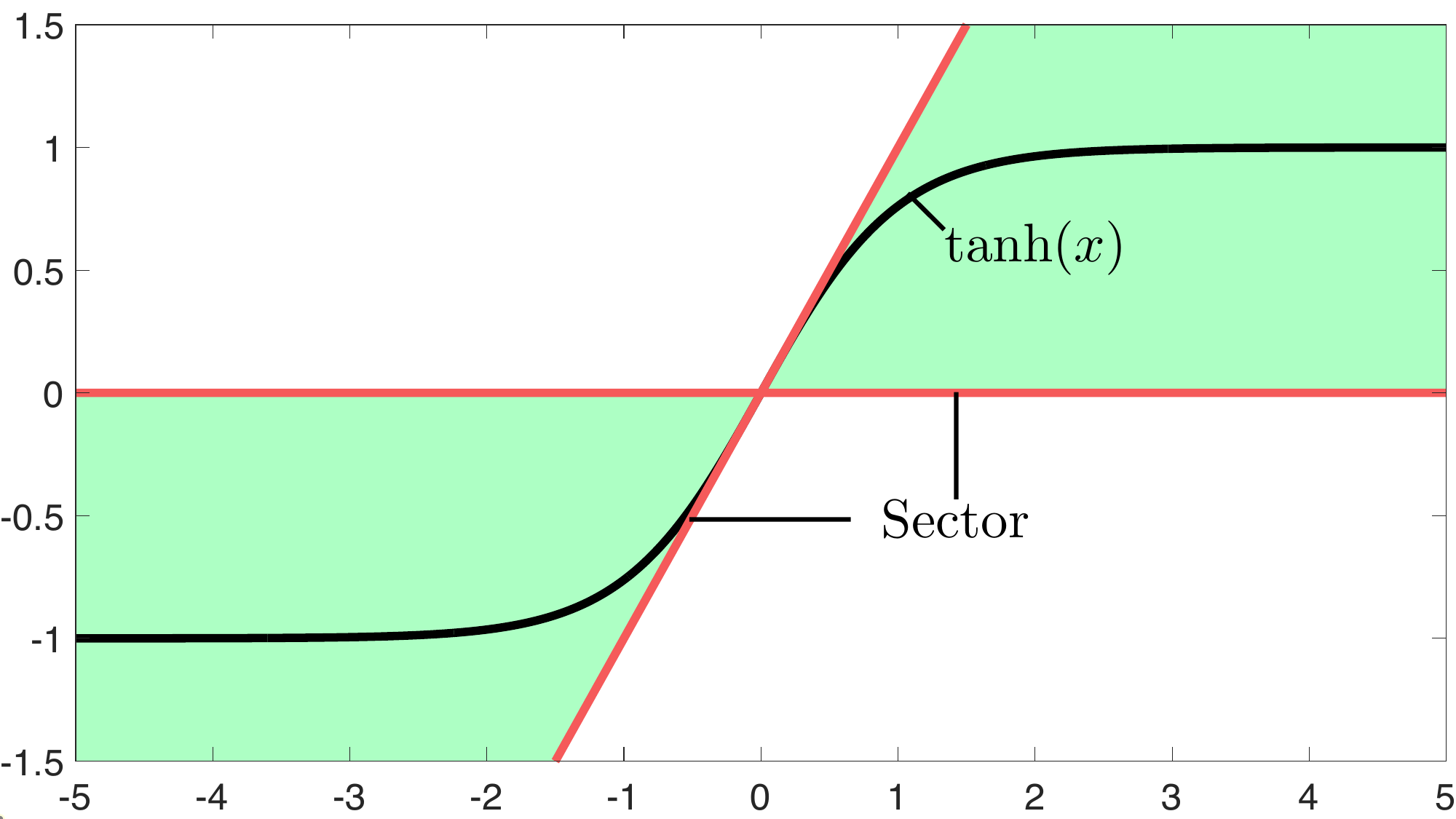}
        \caption{Tanh Activation Function}    
        \label{fig:sectortanh}
    \end{subfigure}
    \caption{Sector constraints that are used for common activation functions when designing neural network controllers.}
    \label{fig:sectorconstraints}
\end{figure}

A root cause of the problems with recent approaches is that they use traditional NN architectures with ReLU, sigmoid and tanh activation functions. These structures have shown to be very effective in many machine learning tasks, however they may not be preferable for use in control systems. By considering a class of NNs that are better aligned with control system techniques, these large convex approximations could be mitigated. Since Sum of Squares (SOS) programming is effective for systems with polynomial and rational functions, we investigate what happens when the NN is built upon them.

\begin{itemize}
    \item Motivated by the success of rational neural networks in machine learning tasks, we propose novel rational activation functions to approximate the traditional sigmoid and tanh activation functions. We then show their effectiveness in analysing neural feedback loops that contain these rational activation functions. 
    \item We note that we could be falling short of the potential expressivity of the neural network by using fixed activation functions. Therefore we consider a general neural network structure which is built upon equations that are convex in the design parameters. We show that a neural network of this form can be expressed in a similar way to that of a feed forward neural network with rational activation functions.
    \item We consider the Lyapunov stability criteria for constrained dynamics systems. We then propose a novel convex Sum of Squares procedure to recover a stabilising controller for a non-linear system through solving a Sum of Squares feasibility test.
    \item This procedure is extended to the generalised rational neural network architecture, to recover a stabilising rational neural network in a convex way. We then adapt the rational neural network architecture to make it more compatible with Sum of Squares programming and allow the stabilising controller to be recovered.
    \item We show for numerous examples that our proposed procedure and neural network architecture are able to effectively recover stabilising controllers for unstable systems with non-linear plants, noise and parametric uncertainty.
\end{itemize}

\section{Preliminaries}

\subsection{Sum of Squares Programming and The Positivstellensatz} \label{sec:sosprogpsatz}
In this section we outline how to formulate and solve SOS optimisation problems by solving an equivalent Semidefinite Program (SDP). For more details on SOS programming the reader is referred to \cite{apap1,apap2,aahm1,yzhe3}. 

The fundamental idea behind SOS is to replace a polynomial positivity condition with a condition that enforces that the polynomial is a Sum of Squares. 
\begin{definition} \label{def:sos}
    A polynomial $p(x)$ is said to be a \emph{Sum of Squares (SOS)} polynomial if it can be expressed as 
    \begin{equation*}
        p(x) = \sum_{i=1}r_{i}^{2}(x), 
    \end{equation*}
    where $r_i(x) \in \mathbb{R}[x]$. We denote the set of polynomials that admit this decomposition by $\Sigma[x]$ and say `$p(x)$ is SOS'. Note that $\mathbb{R}[x_{1}, \dots , x_{n}]$ is defined as the set of polynomials in $x_1, \ldots, x_n$ with real coefficients and we denote $x = (x_1, \ldots, x_n)$ for simplicity. 
\end{definition}
\begin{definition}
    A \emph{monomial} in $x = (x_{1}, \dots x_{n})$ is  
    \begin{equation*}
        x^{\beta} = x_{1}^{\beta_{1}}x_{2}^{\beta_{2}} \dots x_{n}^{\beta_{n}},
    \end{equation*}
    where the exponent and degree are denoted as $\beta = (\beta_{1}, \dots , \beta_{n}) \in \mathbb{N}^{n}$ and $|\beta| = \beta_{1} + \dots + \beta_{n}$ respectively.
\end{definition}
The column vector of monomials with only certain exponents is expressed as $x^{\mathbb{B}} = (x^{\beta})_{\beta \in \mathbb{B}}$, where $\mathbb{B} \subset \mathbb{N}^{n}$ is the set of exponents that are used in the monomials. The summation operation on $\mathbb{B}$ is defined as 
\begin{equation*}
    \mathbb{B} + \mathbb{B} := \{ \beta + \gamma \: : \: \beta, \gamma \in \mathbb{B} \}.
\end{equation*}
A polynomial $p(x)$ can be written as a linear combination of a set of monomials $x^{\mathbb{B}}$ for a set of coefficients $p_{\beta} \in \mathbb{R}$ such that
\begin{equation*}
     p = \sum_{\beta \in \mathbb{N}_{d}^{n}}p_{\beta} x^{\beta},
\end{equation*}
where $\mathbb{N}_{d}^{n} = \{ \beta \in \mathbb{N}^{n} \: : \: |\beta| \leq d \}$ is the set of all $n$-variate exponents of degree $d$ or less.
\begin{theorem}
A polynomial $p(x)$ is SOS if and only if it can be written in what is referred to as a \emph{Gram matrix representation} such that 
\begin{equation*}
    p(x) = (x^{\mathbb{B}})^{T} Q x^{\mathbb{B}},
\end{equation*} 
where $Q \in \mathbb{S}_{+}^{|\mathbb{B}|}$ is a positive semidefinite matrix.
\end{theorem}
The Gram matrix representation can be rewritten as
\begin{equation*}
    (x^{\mathbb{B}})^{T} Q x^{\mathbb{B}} = \langle Q, x^{\mathbb{B}} (x^{\mathbb{B}})^{T} \rangle = \sum_{\alpha = \mathbb{B} + \mathbb{B}} \langle Q, A_{\alpha} \rangle x^{\alpha},
\end{equation*}
where the symmetric binary matrix $A_{\alpha} \in \mathbb{S}^{|\mathbb{B}|}$ for each exponent $\alpha \in \mathbb{B} + \mathbb{B}$ is
\begin{equation*}
    [A_{\alpha}]_{\beta,\gamma} := \begin{cases}
    1, \: \beta + \gamma = \alpha, \\
    0, \: \mathrm{otherwise}.
    \end{cases}
\end{equation*}

The following is therefore true
\begin{equation*}
    p(x) \in \Sigma[x] \Leftrightarrow 
    \exists Q \in \mathbb{S}_{+}^{|\mathbb{B}|} \: \text{where} \:  \langle Q, A_{\alpha} \rangle =p_{\alpha}, \: \forall \alpha \in \mathbb{B} + \mathbb{B}.
\end{equation*}
This means that checking the SOS condition is equivalent to solving an SDP, which can be achieved with parsers such as SOSTOOLS \cite{sostools} in MATLAB. 

A central theorem of real algebraic geometry is known as the Positivstellensatz (Psatz) \cite{gsten1}, which will now be briefly outlined. The Psatz provides an equivalent relation between an algebraic condition and the emptiness of a semi-algebraic set.

We express the semi-algebraic set with notation
\begin{equation} 
    S = \big\{ x \in \mathbb{R}^{n} \: | \:  g_{i}(x) \geq 0, \: h_{j}(x) = 0, \: \forall \: i = 1, \dots, q_1, \: j = 1, \dots , q_2 \big\}, \label{Sset}
\end{equation}
where $g_{i}$ and $h_{j}$ are polynomial functions. 
\begin{definition} \label{def:monoid}
    Given $f_{1}, \dots , f_{r} \in \mathbb{R}[x]$, the \emph{multiplicative monoid} generated by the $f_{k}$'s is the set of all finite products of $f_{k}$'s, including 1 (i.e. the empty product). It is denoted as $\mathcal{M}(f_{1}, \dots , f_{r})$.
\end{definition}
\begin{definition} \label{def:cone}
    Given $g_{1}, \dots , g_{q_1} \in \mathbb{R}[x]$, the \emph{cone} generated by the $g_{i}$'s is 
    \begin{equation} \label{eq:cone}
        \mathrm{cone}\{g_{1}, \dots , g_{q_1} \} =  \Bigg\{ s_{0} + \sum_{i=1}^{q_1} s_{i}G_{i} \: | \:  s_{i} \in \Sigma [x], \: G_{i} \in \mathcal{M}(g_{1}, \dots , g_{q_1}) \Bigg\}.
    \end{equation}
\end{definition}
\begin{definition} \label{def:ideal}
    Given $h_{1}, \dots , h_{q_2} \in \mathbb{R}[x]$, the \emph{ideal} generated by the $h_{k}$'s is
    \begin{equation} \label{eq:ideal}
        \mathrm{ideal}\{h_{1}, \dots , h_{q_2} \} = \Bigg\{ \sum_{j=1}^{q_2} t_{j}h_{j} \: | \: t_{j} \in \mathbb{R}[x] \Bigg\}.
    \end{equation}
\end{definition}
\begin{theorem} (Positivstellensatz, \cite{gsten1}) \label{psatz1}
    Given the semi-algebraic set $S$ in \eqref{Sset}, the following are equivalent:
    \begin{enumerate}
        \item The set $S$ is empty.
        \item There exist $s_i \in \Sigma[x]$ in \eqref{eq:cone} and $t_j \in \mathbb{R}[x]$ in \eqref{eq:ideal} such that \[\mathrm{cone}\{g_{1}, \dots , g_{q_1} \} + \mathrm{ideal} \{h_{1}, \dots , h_{q_2} \} = 0.\]
    \end{enumerate}
\end{theorem}

Based on Theorem \ref{psatz1} one can create a convex test for computational purposes by using a representation of the function $p(x)$.
\begin{proposition}
Consider the set $S$ in \eqref{Sset}. If
\begin{equation}
    p = 1 + \sum_{j = 1}^{q_2}t_{j}h_{j} + s_{0} + \sum_{i = 1}^{q_1}s_{i}g_{i}, \label{eq:psatz}
\end{equation}
where $s_{i} \in \Sigma[x]$ and $t_{j}\in \mathbb{R}[x]$, then $p(x)>0, \: \forall x \in S$.
\end{proposition}
To test if the polynomial $p(x) \geq 0$, $\forall x \in S$ using SOS programming we can rewrite \eqref{eq:psatz} as 
\begin{equation*}
     p + \sum_{j = 1}^{q_2}t_{j}h_{j} - \sum_{i = 1}^{q_1}s_{i}g_{i} \in \Sigma[x],
\end{equation*}
where $s_{i} \in \Sigma[x]$ and $t_{j}\in \mathbb{R}[x]$. By selecting higher degree multipliers $s_i$, $t_j$ etc. we can obtain a series of set emptiness tests with increasing complexity and non-decreasing accuracy.

\subsection{Stability of Neural Feedback Loops using Sum of Squares} \label{sec:stabilityNFLs}
In this section we outline the methods proposed in \cite{mnew5}, which presents a method to determine if the equilibrium a closed loop NFL is stable and then also use this to compute an inner approximation of the region of attraction. 

Consider a continuous-time system
\begin{align} \label{eq:nfl2}
\begin{split}
    \dot{z}(t) &= f(z(t),u(t)),
\end{split}
\end{align}
where $f$ is the plant model $z(t) \in \mathbb{R}^{n_{z}}$ and $u(t) \in \mathbb{R}^{n_{u}}$ are the system states and inputs respectively. $n_{z}$ and $n_{u}$ are the number of system states and inputs respectively. This system is a continuous time NFL if the controller $u$ is an NN. Consider a state feedback controller $u(t) = \pi(z(t)):  \mathbb{R}^{n_{z}} \rightarrow \mathbb{R}^{n_{u}}$ as a feed-forward fully connected NN such that
\begin{align} \label{eq:nnconstab}
\begin{split}
    x^{0}(t) &= z(t), \\ 
    v^{k}(t) &= W^{k}x^{k}(t) + b^{k}, \: \mathrm{for} \: k = 0,\dots, \ell - 1, \\ x^{k+1}(t) &= \phi (v^{k}(t)), \: \mathrm{for} \: k = 0,\dots, \ell - 1, \\
     \pi(z(t)) &= W^{\ell}x^{\ell}(t) + b^{\ell}, 
    \end{split}
\end{align}
where $W^{k} \in \mathbb{R}^{n_{k+1} \times n_{k}}$, $b^{k} \in \mathbb{R}^{n_{k+1}}$ are the weights matrix and biases of the $(k+1)^{\text{th}}$ layer respectively and $z(t) = x^{0}(t) \in \mathbb{R}^{n_{z}}$ is the input into the NN. The activation function $\phi$ is applied element-wise to the $v^{k}(t)$ terms. The number of neurons in the $k^{\text{th}}$ layer is denoted by $n_{k}$. 

As shown in \cite{mnew5}, an NFL can be viewed as a dynamical system with equality and inequality constraints that arise from the input-output description of the NN. These constraints can be described as the semi-algebraic set in \eqref{Sset}. Further constraints can be added to the semi-algebraic set when only local asymptotic stability is being verified. The region 
\begin{equation}
    D^{z} = \left \{ z \in \mathbb{R}^{n_z} \: | \: d_{k}(z) \geq 0, \: k=1,\dots, n_{d} \right \}, \label{eq:Dz}
\end{equation}
is considered where the stability conditions will need to be satisfied. 

\begin{proposition} (\cite{mnew5}) \label{prop:nflstability}
Consider System \eqref{eq:nfl2} in feedback with an NN controller given by \eqref{eq:nnconstab}. Suppose the input-output properties of the NN are described by \eqref{Sset}, and consider the region given by \eqref{eq:Dz}. Suppose there exists a polynomial function $V(z)$ satisfying the following conditions
\begin{align} \label{sosopt2}
\begin{split}
    V(z) - \rho(z) &\in \Sigma[z], \\
    \rho(z) &> 0, \\
    -\frac{\partial V}{\partial z}(z) f(z,\pi(z)) - \sum_{k=1}^{n_d}p_{k}(X)d_{k}(z) -& \sum_{j=1}^{q_2}t_{j}(X)h_{j}(x)  - \sum_{i=1}^{q_1}s_{i}(X)g_{i}(x) \in \Sigma[X], \\
    p_{k}(X) &\in \Sigma[X], \: \forall k = 1, \dots , n_{d}, \\
    s_{i}(X) &\in \Sigma[X], \: \forall i = 1, \dots, q_{1}, \\
    t_{j}(X) &\in \mathbb{R}[X], \: \forall j = 1, \dots, q_{2}, 
\end{split}
\end{align}
where $X$ is a vector of all the system and NN states, i.e. $X = (x,z)$. Then the equilibrium of the NFL is stable.
\end{proposition}

If a Lyapunov function is constructed using Proposition \ref{prop:nflstability} then the region of attraction can be approximated. To achieve this we find the largest level set of the Lyapunov function $V(z)$ that is contained within the region that the Lyapunov conditions are satisfied. This can be cast as an SOS program. Consider a Lyapunov function $V(z)$, if the SOS optimisation problem
\begin{align} \label{eq:sosroa}
\begin{split}
    |z|^k(V(z) - \gamma) + p(z)d(z)  &\in \Sigma[z], \\
    p(z) &\in \Sigma[z], 
\end{split}
\end{align}
is feasible, where $\gamma$ is a variable to be maximised and $k$ is a positive integer, then $V(z) \leq \gamma$ is an estimate of the region of attraction. 

\section{Rational Neural Networks} \label{sec:rationalapproximations}
There has been little prior work investigating NNs with rational expressions for control. \cite{mtel} showed the expressive power of using rational activation functions and how they can be used to approximate commonly used activation functions such as ReLU. The effectiveness of rational NNs in machine learning tasks was shown in \cite{nbou}. Polynomial activation functions have also been used in previous work \cite{lhou,yber,elop,ffar,boli}, however they have not seen much investigation within the machine learning research community due to vanishing and exploding gradients that can be exhibited when used with the backpropagation algorithm \cite{sdub}. Another class of NNs that has seen recent interest are quadratic NNs. These can take many different forms, which are summarised in \cite{zixu}. Quadratic NNs have been shown to be beneficial in many aspects \cite{sidu,fefa,ibuk}, such as being general universal function approximators. Recent work from \cite{lrod} has explored the use of a two-layer NN controller that has quadratic activation functions and how a convex formulation can be achieved using this structure. 

One of the benefits of using polynomial or rational activation functions in the NN is that they can contain trainable parameters, which can improve the performance of the NN and reduce the number of neurons in the network. For example, the rational activation function used in \cite{nbou} is expressed as
\begin{equation} \label{eq:simplerationalAF}
    \phi(x) = \frac{P(x)}{Q(x)} = \frac{\sum_{i=0}^{r_{P}}a_{i}x^{i}}{\sum_{j=0}^{r_{Q}}b_{j}x^{j}},
\end{equation}
where $r_{P} = \mathrm{deg}(P(x))$, $r_{Q} = \mathrm{deg}(Q(x))$ and $a_{i}$ and $b_{j}$ are the trainable parameters within the activation function. Another proposed activation function based on rational functions is the Pad\'e activation unit, which has shown to be useful when applied to image classification \cite{amol}. This activation function is expressed as
\begin{equation*}
    \phi(x) = \frac{P(x)}{Q(x)} = \frac{\sum_{i=0}^{r_{P}}a_{i}x^{i}}{1+ | \sum_{j=0}^{r_{Q}}b_{j}x^{j} |}
\end{equation*}
and can be used to learn and approximate commonly used activation functions, whilst training the NN. The resulting NNs can have compact representations and can perform similarly to state-of-the-art NNs.

\subsection{Rational Approximation of Tanh Activation Function} \label{sec:tanhrationalapproximations}
Despite their success in some machine learning applications, rational activation functions have yet to be applied to NN controllers. Most methods to train NNs in control systems use traditional activation functions such as ReLU, sigmoid and tanh. In addition, NNs that are used in control systems are often significantly smaller than those used for machine learning tasks such as image classification. The reason behind this is that methods to train NNs for control often require the use of an SDP, meaning that having a large number of parameters in the NN makes the problem become intractable. However, the requirement for the NN to only contain a small number of neurons to learn a function that can sufficiently control a system has not been well explored. One justification for this is the low number of input and output dimensions that are required for small NFLs. However, if the number of system states increases, then larger networks may be required, which may be intractable to obtain with current methods. 

For NNs to be effectively used as controllers, it would be of interest to investigate the use of alternative NN structures and activation functions that may give a sufficient level of expressivity in the network, whilst being easier to compute or ensure robustness guarantees. Motivated by this we propose a novel activation function defined by a simple rational expression. We name this function `Rtanh' as it is an approximation of the tanh activation function 
\begin{equation*}
    \mathrm{tanh}(x) = \phi(x) = \frac{e^{x} - e^{-x}}{e^{x} + e^{-x}}
\end{equation*}
and is defined as
\begin{equation*}
    \mathrm{Rtanh}(x) = \phi(x) = \frac{4x}{x^2 + 4}.
\end{equation*}
This function is shown in Figure \ref{fig:RtanhCompare} against the tanh function and the error between these functions is shown in Figure \ref{fig:RtanhDiff}.

\begin{figure}[h!]
    \centering
    \begin{subfigure}[b]{0.49\textwidth}
        \centering
        \includegraphics[width=\textwidth]{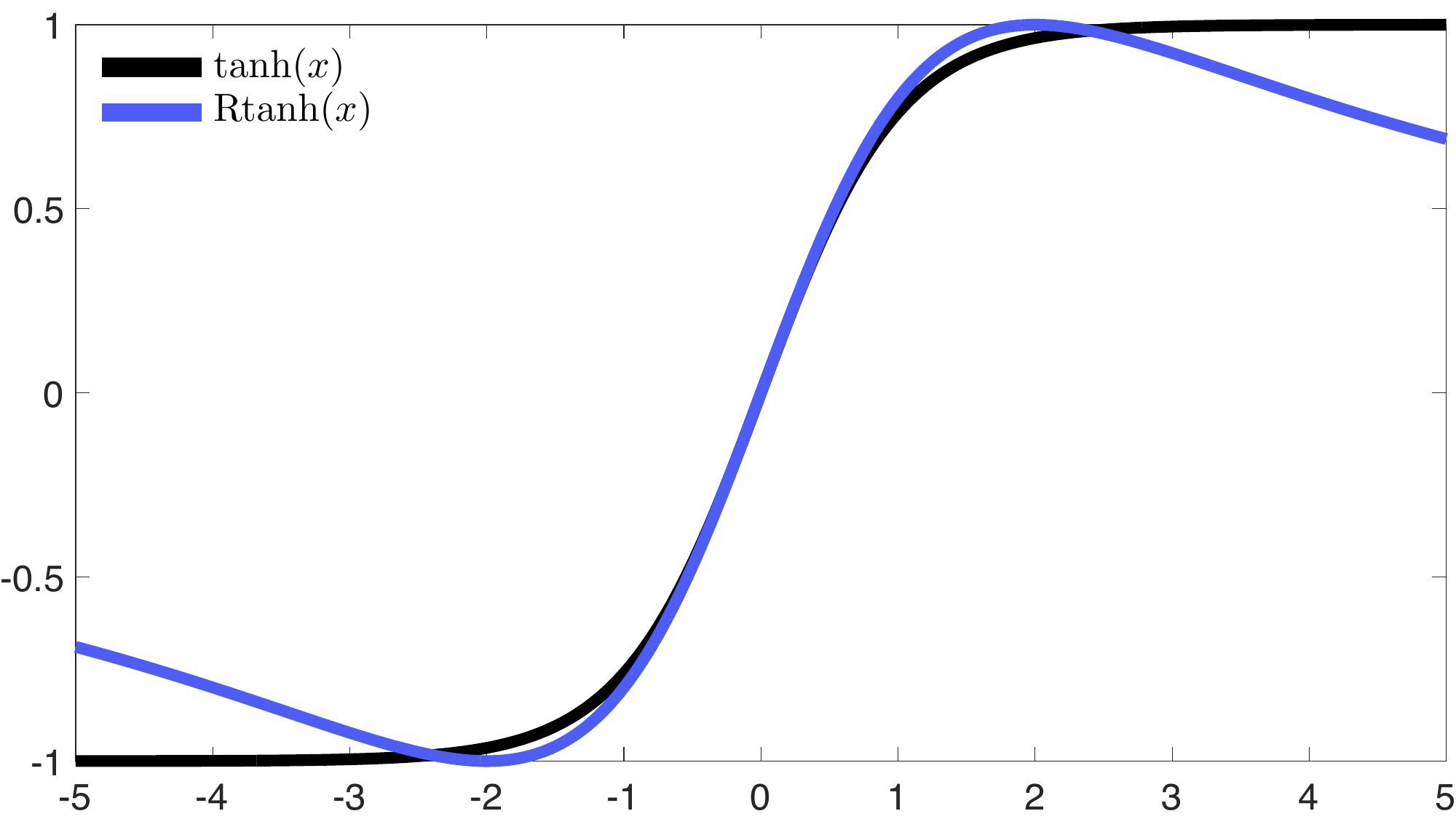} 
        \caption{Comparison between $\mathrm{Rtanh}(x)$ and $\mathrm{tanh}(x)$}
        \label{fig:RtanhCompare}
    \end{subfigure}
    \hfill
    \begin{subfigure}[b]{0.49\textwidth}  
        \centering 
        \includegraphics[width=\textwidth]{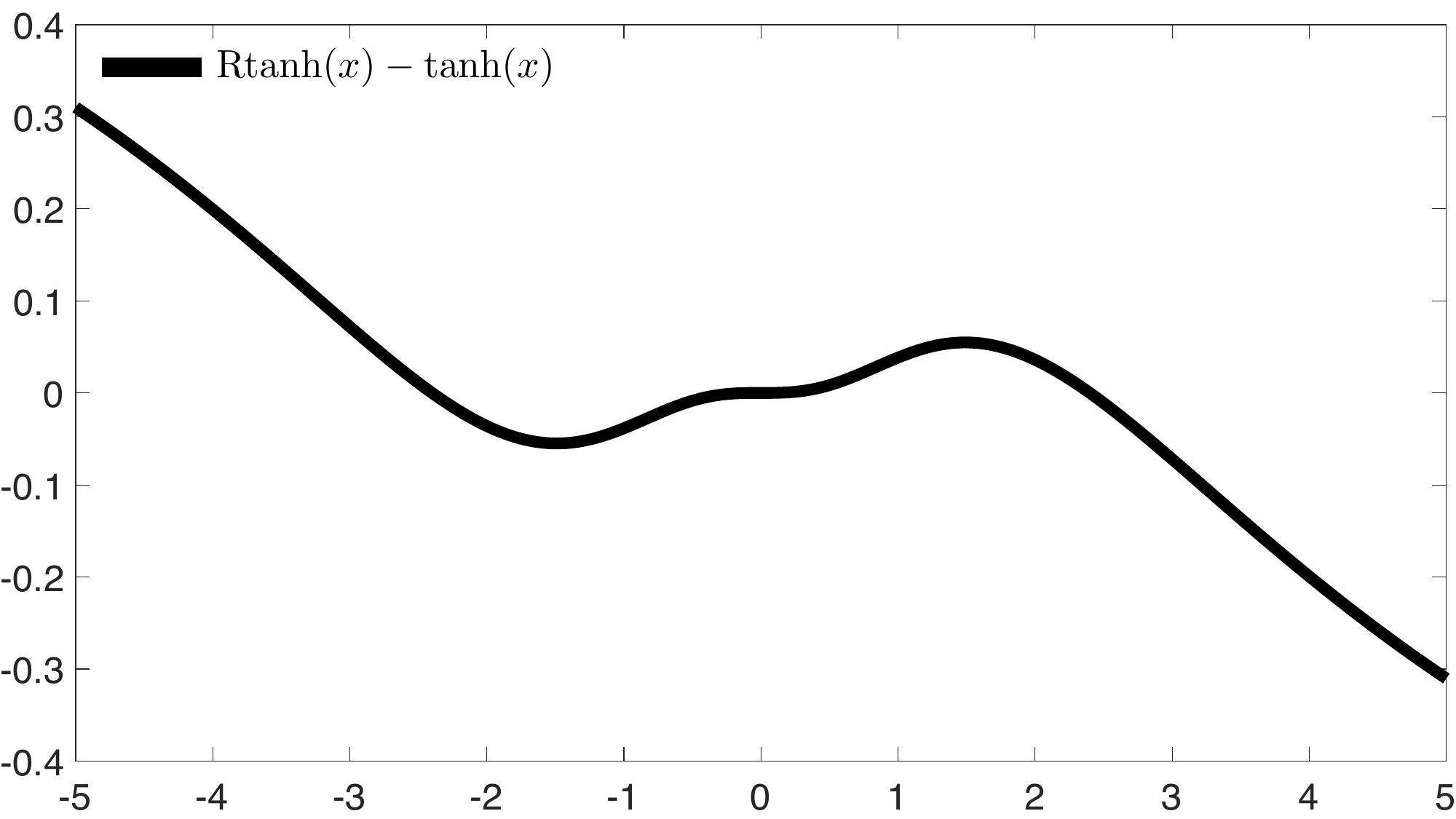}
        \caption{Error between $\mathrm{Rtanh}(x)$ and $\mathrm{tanh}(x)$}
        \label{fig:RtanhDiff}
    \end{subfigure}
    \caption{Plots showing a comparison between the $\mathrm{Rtanh}(x)$ and $\mathrm{tanh}(x)$ functions.}
\end{figure}

The tanh function can be over-approximated by sector constraints, as demonstrated in \cite{hyin2,fagu,njun}. These constraints can then be used to analyse the NFL by creating an optimisation problem through an SDP or SOS programming. However, these bounds are very conservative as shown in Figure \ref{fig:sectortanh}. A big advantage of the `Rtanh' function is that it can be represented by a single equality constraint such that
\begin{equation} \label{eq:rtanheq}
    \phi(x)(x^2 + 4) - 4x = 0.
\end{equation}
If this activation function were to be used in an NN, then to test its robustness properties using the Psatz, the equality constraint \eqref{eq:rtanheq} can be used directly. This is beneficial as no conservatism is introduced when abstracting the input-output properties of the NN using a semi-algebraic set. 

To demonstrate that the Rtanh activation function is useful when analysing NFLs, we take an existing NFL that uses tanh activation functions. We consider the inverted pendulum from \cite{ppau2}, which uses a five layer NN with five nodes in each layer and tanh activation functions. The dynamics of this system are expressed as
\begin{equation} \label{eq:ipreachrational} 
    \ddot{\theta}(t) = \frac{mgl \sin{\theta (t)} - \mu \dot{\theta} (t) + \text{sat}(u(t))}{m l^{2}}, 
\end{equation}
where $\text{sat}(\cdot)$ is the saturation function. The system is discretised with time step $\Delta t = 0.2$ and is parameterised by $m = 0.15$kg, $l=0.5$m, $\mu = 0.5$Nms/rad, $g=10$m/s\textsuperscript{2}, $u_{\text{max}} = 1$ where $u_{\text{max}}$ is the saturation limit in the saturation function. As shown in \cite{mnew3}, we can use `ReachSparsePsatz', which is an SOS optimisation technique to approximate the reachable set at each time step. Using the tanh activation function as in the original NN controller, the reachable sets can be computed and are shown in Figure \ref{fig:ReachTanhAF}. 

\begin{figure}[h!] 
    \centering  
    \includegraphics[height=8cm]{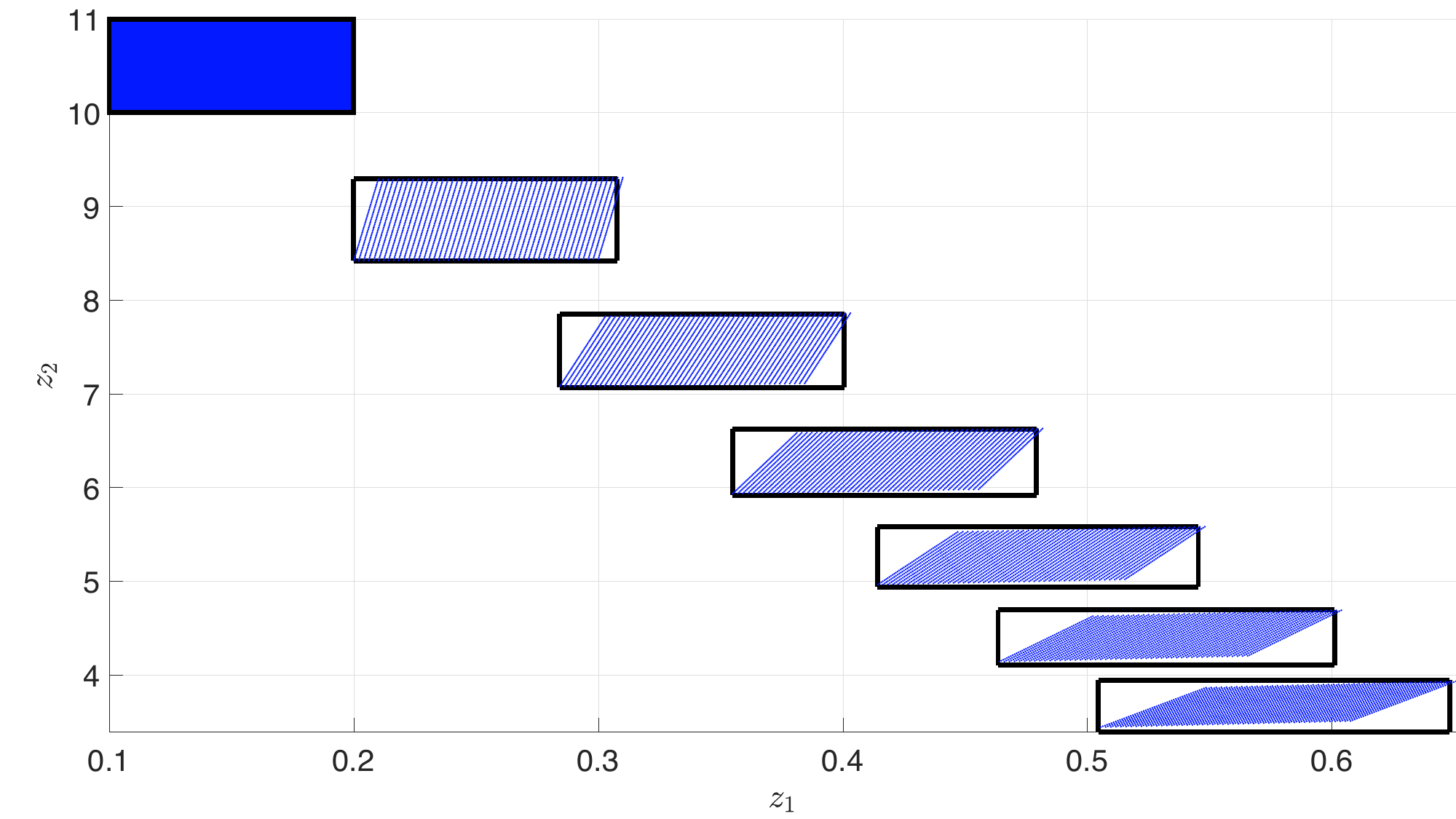}
    \caption{Reachable sets for the inverted pendulum in \eqref{eq:ipreachrational} with the tanh activation function. The reachable sets are computed using ReachSparsePsatz and are shown in black. The true reachable sets are shown in blue.} \label{fig:ReachTanhAF}
\end{figure}

We then replace the tanh function with Rtanh and observe the system behaviour. We conduct the same reachability analysis as in \cite{mnew3} by computing the reachable sets for six time steps. In Figure \ref{fig:ReachRationalAF} we can see that this activation function behaves similarly to that of the tanh function and that using the Psatz with sparse polynomial optimisation (ReachSparsePsatz) gives very tight approximations of the reachable sets. 

\begin{figure}[h!] 
    \centering  
    \includegraphics[height=8cm]{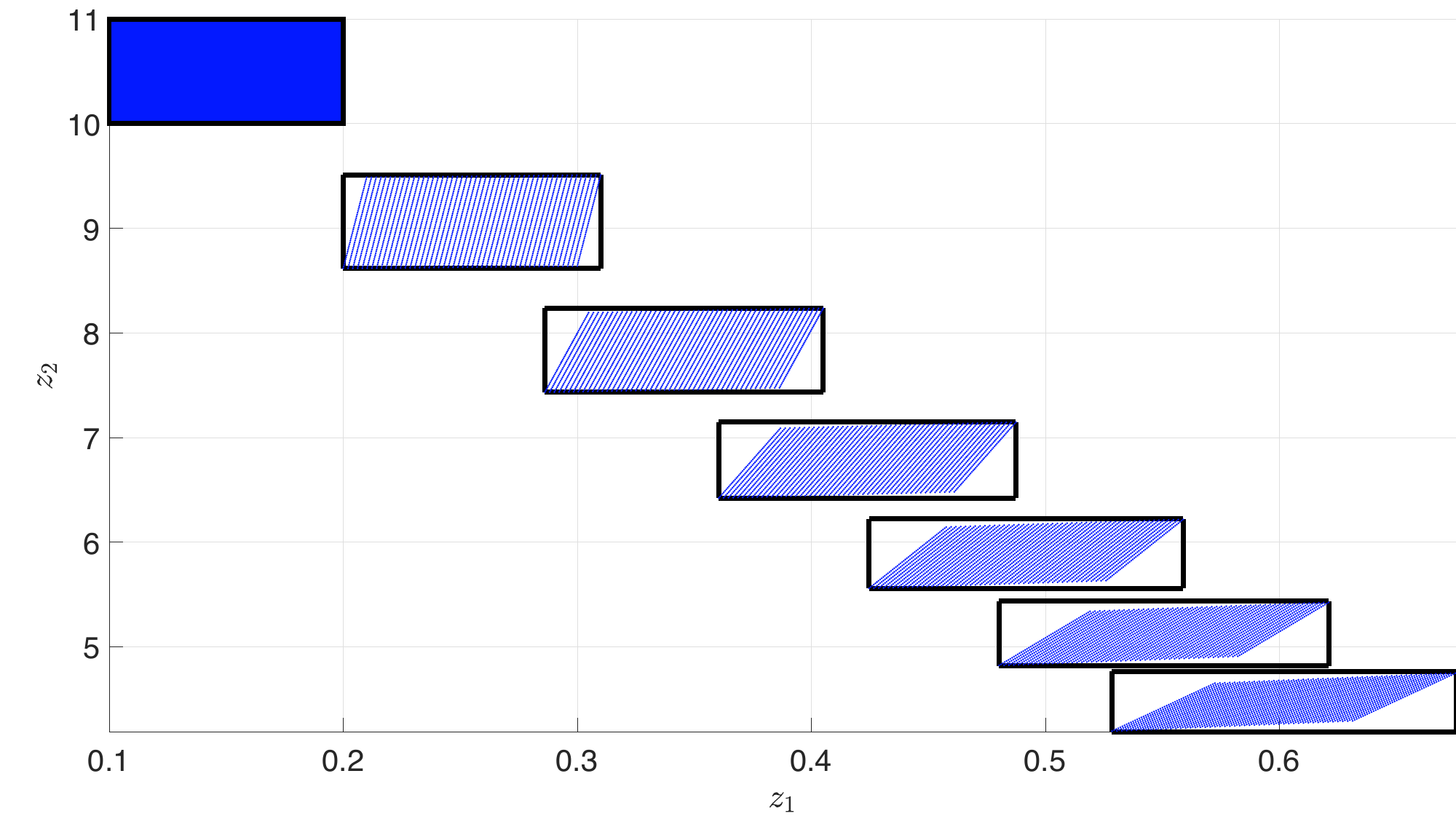}
    \caption{Reachable sets for the inverted pendulum in \eqref{eq:ipreachrational} where the tanh activation function is replaced with the Rtanh function. The reachable sets are computed using ReachSparsePsatz and are shown in black. The true reachable sets are shown in blue.} \label{fig:ReachRationalAF}
\end{figure}

We also compute the region of attraction using the approach outlined in Section \ref{sec:stabilityNFLs}, refereed to as `NNSOSStability'. The region of attraction when using the tanh and Rtanh functions are shown in Figure \ref{fig:ROATanhAF} and Figure \ref{fig:ROARationalAF} respectively. We find that the region of attraction is increased significantly when using Rtanh over the tanh activation function. The areas of the region of attraction on the phase plane are 58 and 0.94 per square unit for the Rtanh and tanh activation functions respectively. 

\begin{figure}[h!] 
    \centering  
    \includegraphics[height=8cm]{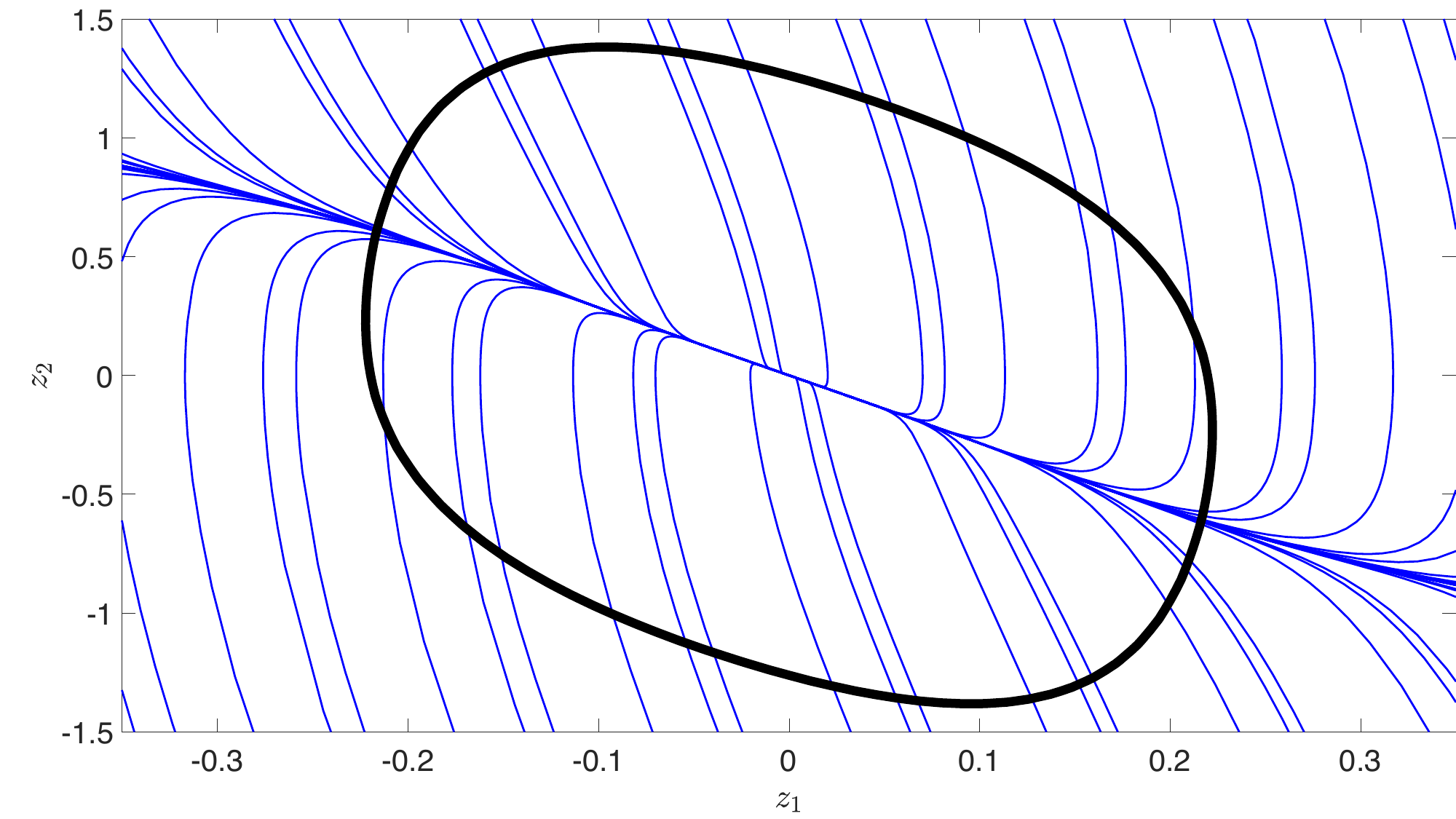}
    \caption{Diagram showing the region of attraction for the inverted pendulum in \eqref{eq:ipreachrational} with tanh activation functions in the neural network controller. The region of attraction is computed using NNSOSStability (black) with a fourth order Lyapunov function. The system trajectories are shown in blue.} \label{fig:ROATanhAF}
\end{figure}

\begin{figure}[h!] 
    \centering  
    \includegraphics[height=8cm]{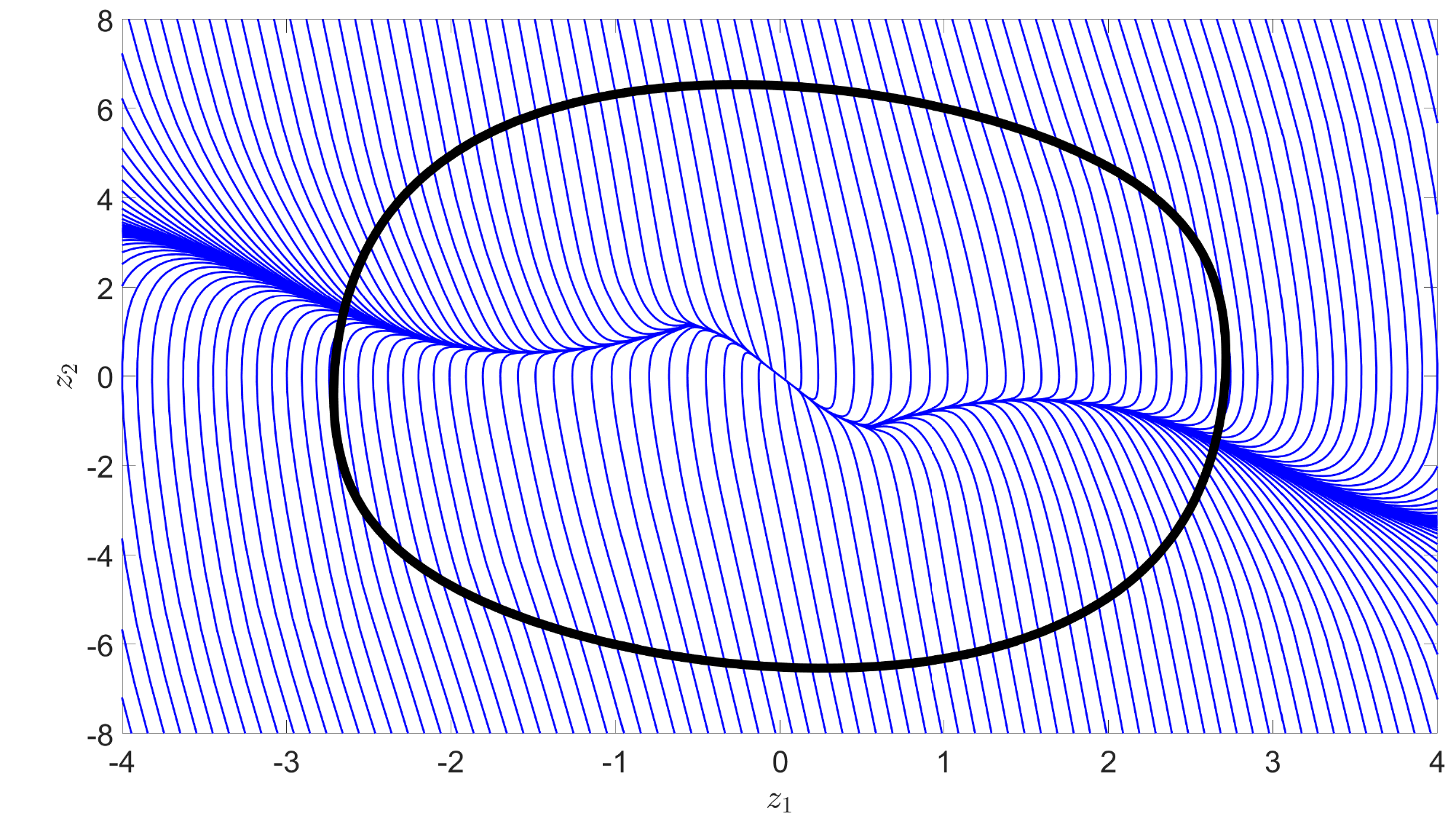}
    \caption{Diagram showing the region of attraction for the inverted pendulum in \eqref{eq:ipreachrational} where the tanh activation function is replaced with the Rtanh function. The region of attraction is computed using NNSOSStability (black) with a fourth order Lyapunov function. The system trajectories are shown in blue.} \label{fig:ROARationalAF}
\end{figure}

These results show that the Rtanh function can be useful in NNs and NFLs. However, this function is not compatible with most recent methods to learn NN controllers due to it being represented by an equality that is a polynomial of degree three. Indeed, most methods require sector inequality constraints and the Schur complement, which is used to make the problem convex. Therefore, this function cannot be used in those formulations. This requires the use of an alternative method to obtain an NN controller with these activation functions, which we will investigate later in this paper.

\subsection{Rational Approximation of Sigmoid Activation Function}
It is also possible to approximate the sigmoid function as a rational function. We define the function
\begin{equation*}
    \mathrm{Rsig}(x) = \phi (x) = \frac{(x + 4)^2}{2(x^{2} + 16)}. 
\end{equation*}
As shown in Figure \ref{fig:RsigCompare}, this is a good approximation to the sigmoid function. Figure \ref{fig:RsigDiff} shows the error with the sigmoid function. Rsig can be expressed as the equality constraint
\begin{equation*}
    2(x^{2} + 16)\phi (x) - (x + 4)^2 = 0.
\end{equation*}

\begin{figure}[h!]
    \centering
    \begin{subfigure}[b]{0.49\textwidth}
        \centering
        \includegraphics[width=\textwidth]{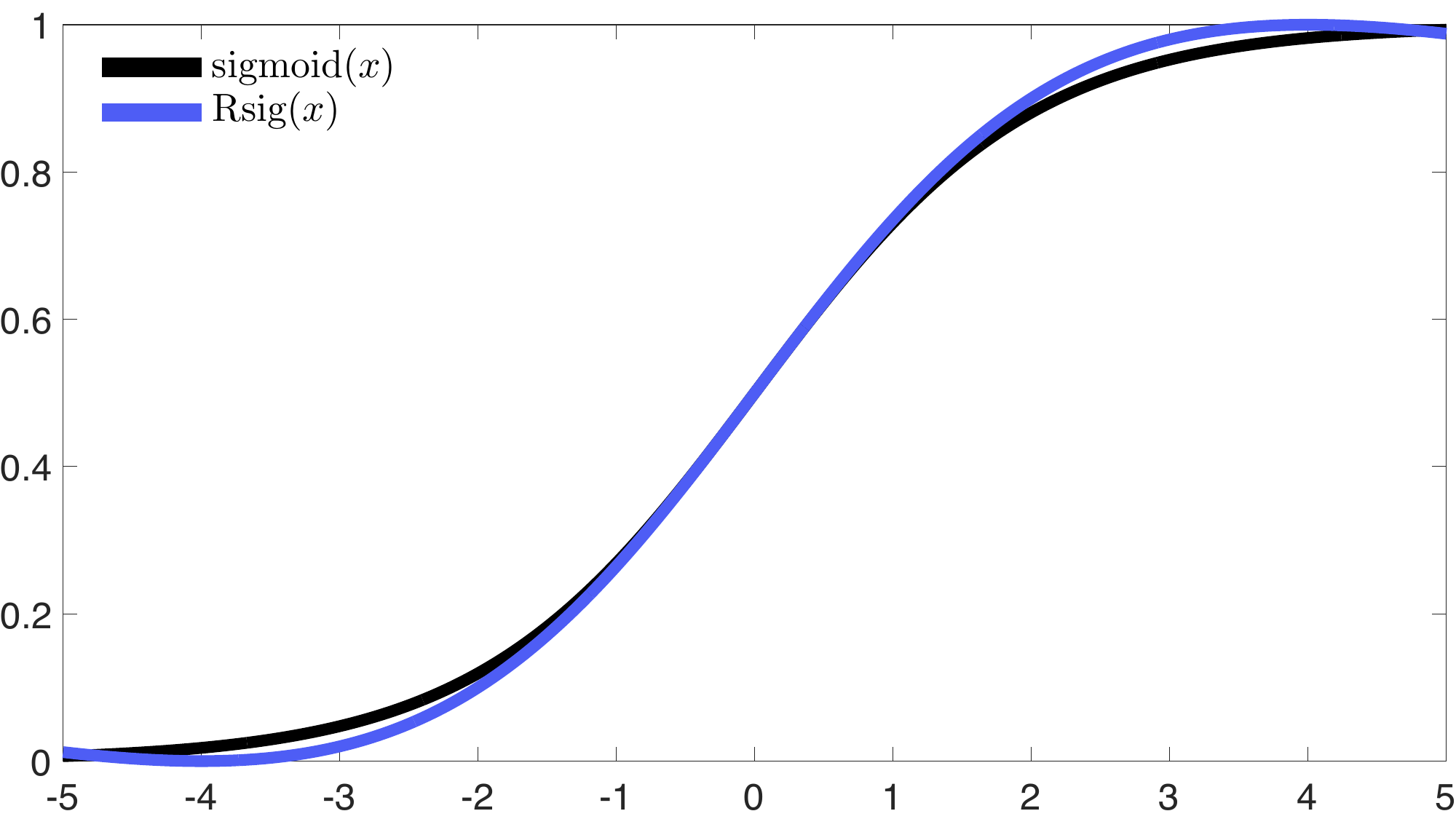} 
        \caption{Comparison between $\mathrm{Rsig}(x)$ and $\mathrm{sig}(x)$}
        \label{fig:RsigCompare}
    \end{subfigure}
    \hfill
    \begin{subfigure}[b]{0.49\textwidth}  
        \centering 
        \includegraphics[width=\textwidth]{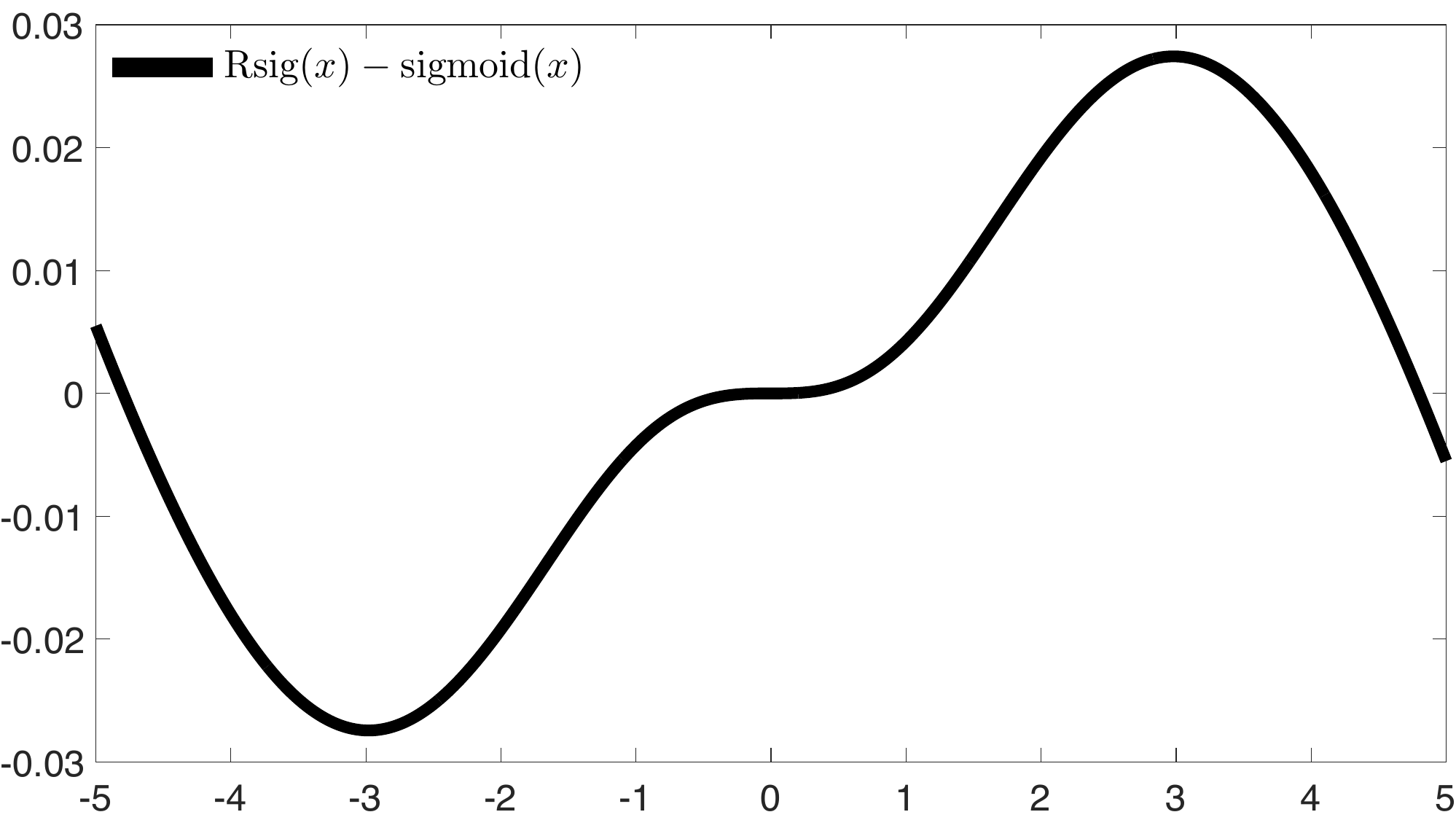}
        \caption{Error between $\mathrm{Rsig}(x)$ and $\mathrm{sig}(x)$}
        \label{fig:RsigDiff}
    \end{subfigure}
    \caption{Plots showing a comparison between the $\mathrm{Rsig}(x)$ and $\mathrm{sig}(x)$ functions.}
\end{figure}

\subsection{Irrational Approximation of ReLU Activation Function}
Approximating the ReLU function with a rational function is more difficult. Here we consider the class of irrational activation functions; to demonstrate how they could be used we propose the following function
\begin{equation*}
    \mathrm{IReLU}(x) = \phi(x) = \sqrt{x^2 + 1} + x - 1.
\end{equation*}
This function and the error with the ReLU function are shown in Figure \ref{fig:IReLUCompare} and Figure \ref{fig:IReLUDiff} respectively. By making a substitution, the IReLU function can be expressed as a semi-algebraic set with two equality constraints and one inequality constraint
\begin{align*}
    \phi(x) - y - x + 1 &= 0, \\
    y^2 - x^2 - 1 &= 0, \\
    y &\geq 0.
\end{align*}
\begin{figure}[h!]
    \centering
    \begin{subfigure}[b]{0.49\textwidth}
        \centering
        \includegraphics[width=\textwidth]{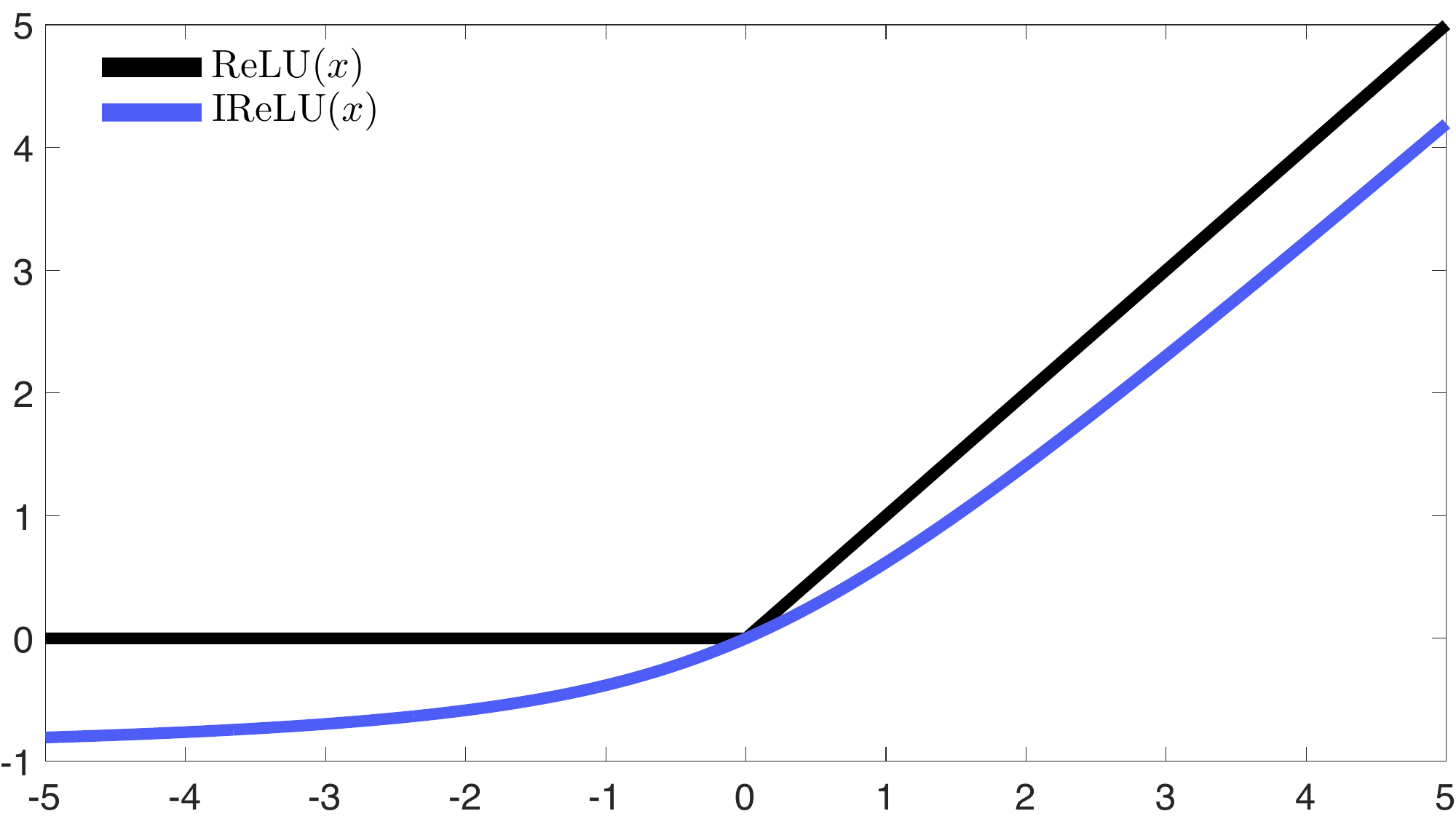} 
        \caption{Comparison between $\mathrm{IReLU}(x)$ and $\mathrm{ReLU}(x)$}
        \label{fig:IReLUCompare}
    \end{subfigure}
    \hfill
    \begin{subfigure}[b]{0.49\textwidth}  
        \centering 
        \includegraphics[width=\textwidth]{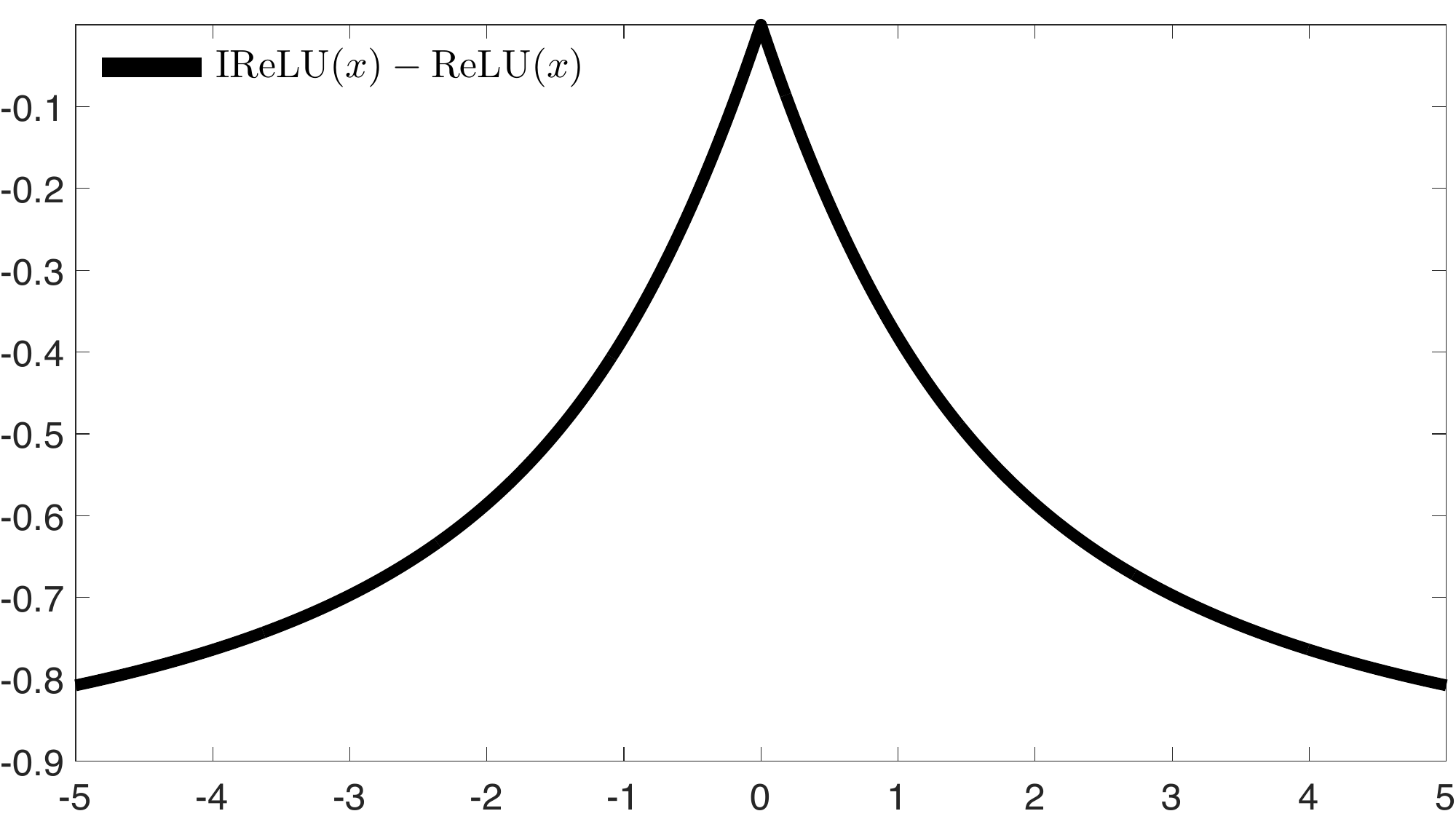}
        \caption{Error between $\mathrm{IReLU}(x)$ and $\mathrm{ReLU}(x)$}
        \label{fig:IReLUDiff}
    \end{subfigure}
    \caption{Plots showing a comparison between the $\mathrm{IReLU}(x)$ and $\mathrm{ReLU}(x)$ functions.}
\end{figure}

\subsection{General Rational Neural Networks}
The rational approximations of the tanh and sigmoid functions proposed in Section \ref{sec:rationalapproximations} can be used as activation functions in an NN. However, if we were to use these structures we could be missing the potential expressivity that can be obtained from a general class of rational functions. As in \eqref{eq:simplerationalAF}, rational activation functions can contain training parameters. Therefore, instead of considering a rigid structure for the activation function and fixing the parameters, we can consider a general rational expression that contains the NN parameters.

If we were to consider a feed-forward fully connected NN with predefined rational activation functions, then the NN can be written as 
\begin{equation} \label{eq:oldrationalnn}
\begin{aligned} 
    x^{0} &= u, \\ 
    v^{k} &= W^{k}x^{k} + b^{k}, \: \mathrm{for} \: k = 0,\dots, \ell - 1, \\ 
    x^{k+1} &= \frac{p(v^{k})}{q(v^{k})}, \: \mathrm{for} \: k = 0,\dots, \ell - 1, \\
    \pi(u) &= W^{\ell}x^{\ell} + b^{\ell}, 
\end{aligned}
\end{equation}
where $p(v^{k})$ and $q(v^{k})$ are polynomial functions with specified coefficients. However, if we substitute the preactivation value $v_{j}^{k}$ into the polynomial expressions, we obtain a rational expression in $x_{j}^{k}$, where the coefficients are parameterised by the values of the weight matrices $W^{k}$ and biases vector $b^{k}$. Any coefficients in the rational activation function will be multiplied by the weights and bias terms. Therefore, we can instead consider an NN with no affine transformation and parameters only in the rational activation function. We can then write the rational NN as
\begin{equation} \label{eq:newrationalnn}
\begin{aligned} 
    x^{0} &= u, \\ 
    x_{i}^{k+1} &= \frac{p_{i}(x^{k})}{q_{i}(x^{k})}, \: \mathrm{for} \: i = 0,\dots, n_{k}, \: \mathrm{for} \: k = 0,\dots, \ell - 1, \\
    \pi_{i}(u) &= \frac{p_{i}(x^{\ell})}{q_{i}(x^{\ell})}, \: \mathrm{for} \: i = 0,\dots, n_{\ell},
\end{aligned}
\end{equation}
where $p_{i}(x^{k})$ and $q_{i}(x^{k})$ are general polynomial functions which can be written as
\begin{align*}
    p_{i}(x) &= \sum_{\alpha \in \mathbb{N}_{d_{p_{i}}}^{n}}\lambda_{\alpha}x^{\alpha}, \\
    q_{i}(x) &= \sum_{\beta \in \mathbb{N}_{d_{q_{i}}}^{n}}\gamma_{\beta}x^{\beta},
\end{align*}
where $d_{q_{i}} = \mathrm{deg}(q_{i})$, $ d_{p_{i}} = \mathrm{deg}(p_{i})$, $\alpha$ and $\beta$ are the exponents that are defined in Section \ref{sec:sosprogpsatz} and $\lambda_{\alpha}$ and $\gamma_{\beta}$ are the coefficients of $p_{i}(x)$ and $q_{i}(x)$ respectively. To show the similarity between \eqref{eq:oldrationalnn} and \eqref{eq:newrationalnn} we present the following simple example.

\begin{example}
    Consider an NN with structure defined by \eqref{eq:oldrationalnn} with two layers and two nodes in each layer. The rational activation functions with $\mathrm{deg}(p) = \mathrm{deg}(q) = 2$ can be written as
    \begin{align*}
        p(v) &= c_{1}v^{2} + c_{2}v + c_{3}, \\
        q(v) &= d_{1}v^{2} + d_{2}v + d_{3}.
    \end{align*} 
    The first node in the second layer can be expressed as
    \begin{align*}
        v_{1}^{1} &= W_{1}^{1}x_{1}^{1} + W_{2}^{1}x_{2}^{1} + b_{1}^{1}, \\
        x_{1}^{2} &= \frac{p(v_{1}^{1})}{q(v_{1}^{1})}.
    \end{align*}
    We can substitute in the preactivation terms into the activation functions to obtain the polynomials
    \begin{multline*}
        p(v_{1}^{1}) = c_{1}(W_{1}^{1})^{2}(x_{1}^{1})^{2} + 2c_{2}W_{1}^{1}W_{2}^{1}(x_{1}^{1}x_{2}^{1}) + c_{3}(W_{2}^{1})^{2}(x_{2}^{1})^{2} + \\ (c_{1}W_{1}^{1} + 2c_{2}W_{1}^{1}b_{1}^{1})(x_{1}^{1}) + (c_{1}W_{2}^{1} + 2c_{2}W_{2}^{1}b_{1}^{1})(x_{2}^{1}) + (c_{1}(b_{1}^{1})^{2} + c_{2}b_{1}^{1} + c_{3}),
    \end{multline*}
    \begin{multline*}
        q(v_{1}^{1}) = d_{1}(W_{1}^{1})^{2}(x_{1}^{1})^{2} + 2d_{2}W_{1}^{1}W_{2}^{1}(x_{1}^{1}x_{2}^{1}) + d_{3}(W_{2}^{1})^{2}(x_{2}^{1})^{2} + \\ (d_{1}W_{1}^{1} + 2d_{2}W_{1}^{1}b_{1}^{1})(x_{1}^{1}) + (d_{1}W_{2}^{1} + 2d_{2}W_{2}^{1}b_{1}^{1})(x_{2}^{1}) + (d_{1}(b_{1}^{1})^{2} + d_{2}b_{1}^{1} + d_{3}),
    \end{multline*}
    which form the rational expression for $x_{1}^{2}$. However, we can instead write the rational expression as
    \begin{equation*}
        x_{1}^{2} = \frac{\lambda_{1}(x_{1}^{1})^{2} + \lambda_{2}x_{1}^{1}x_{2}^{1} + \lambda_{3}(x_{2}^{1})^{2} + \lambda_{4}x_{1}^{1} +  \lambda_{5}x_{2}^{1} + \lambda_{6}}{\gamma_{1}(x_{1}^{1})^{2} + \gamma_{2}x_{1}^{1}x_{2}^{1} + \gamma_{3}(x_{2}^{1})^{2} + \gamma_{4}x_{1}^{1} +  \gamma_{5}x_{2}^{1} + \gamma_{6}},
    \end{equation*}
    which is in the form of \eqref{eq:newrationalnn}. This can be generalised to larger NNs with higher degree polynomials in the rational functions. This approach can reduce the number of parameters in the NN and each polynomial is convex in the decision variables. This allows the parameters in the rational expression to be tuned directly instead of simultaneously tuning the weights, biases and rational activation parameters.
\end{example}

\section{Recovering Stabilising Controllers using Sum of Squares} \label{sec:recovercontSOS}
In this section, we propose a novel procedure to obtain a stabilising controller for a non-linear polynomial system using SOS programming. To do this we leverage the Psatz and exploit its structure to generate a feasibility test for a stabilising controller. 

\begin{proposition} \label{prop:rationalpsatz}
    Consider the polynomial $P(x) \in \mathbb{R}[x]$ and partition $x = [y, z]$ so that $y \in \mathbb{R}^{n}, z \in\mathbb{R}$. Consider the set
    \begin{equation*}
        S = \big\{ y \in \mathbb{R}^{n}, \: z \in\mathbb{R}  \: |\: p_{i}(y,z) - q_{i}(y)z = 0 \: \: \forall \: i = 1, \dots, m \big\}, 
    \end{equation*}
    where $p_{i}(y,z) \in \mathbb{R}[x]$, $q_{i}(z) \in \mathbb{R}[z]$. If 
    \begin{equation} \label{eq:propaltpsatz}
        P(x) - \sum_{i=1}^{m} (p_{i}(y,z) - q_{i}(y)z) \in \Sigma[x], \: q_{i}(y) \neq 0, \: \forall \: i = 1, \dots, m,
    \end{equation}
    then $P(x) \geq 0$ on the set
    \begin{equation*}
        T = \bigg\{ y \in \mathbb{R}^{n}, \: z \in\mathbb{R}  \: \bigg| \: \frac{p_{i}(y,z)}{q_{i}(y)} - z = 0, \: q_{i}(y) \neq 0 \: \:  \forall \: i = 1, \dots, m \bigg\}.
    \end{equation*}
\end{proposition}
\begin{proof}
    Consider the set $S$, if
    \begin{equation*}
        P(x) - \sum_{i=1}^{m} t_{i}(p_{i}(y,z) - q_{i}(y)z) \in \Sigma[x],
    \end{equation*}
    where $t_{i} = 1$, since $t_{i} \in \mathbb{R}[x]$, then $P(x) \geq 0$ on $S$. If we include the condition $q_{i}(y) \neq 0, \: \forall \: i = 1, \dots, m$, then we can rewrite \eqref{eq:propaltpsatz} as
    \begin{equation*} 
        P(x) - \sum_{i=1}^{m} q_{i}(y) \bigg( \frac{p_{i}(y,z)}{q_{i}(y)} - z \bigg) \in \Sigma[x], \: q_{i}(y) \neq 0, \: \forall \: i = 1, \dots, m,
    \end{equation*}
    since $q_{i}(y) \in \mathbb{R}[x]$, then $P(x) \geq 0$ on $T$.
\end{proof}

Proposition \ref{prop:rationalpsatz} considers the Psatz in a particular form to obtain a positivity certificate of a function over a set of rational functions. In essence, this formulation fixes the multipliers in the Psatz to unity and then searches over the polynomials to find a constraint set that satisfies the feasibility test. This is useful if we want to find a constraint set for the nonnegativity of a function, instead of testing that function over a known constraint set. This can be leveraged to find a controller that satisfies the Lyapunov stability conditions. 

\begin{proposition} \label{prop:rationalcontroller}
    Consider the non-linear system
    \begin{equation} \label{eq:systemrational}
    \begin{aligned}
        \dot{z} &= f(z) + g(z)u, \\
        u &= \frac{p(z)}{q(z)}, \: q(z) \neq 0,
    \end{aligned}
    \end{equation}
    where $z \in \mathbb{R}^{n_{z}}$ are the system states, $u \in \mathbb{R}^{n_{u}}$ is the controller input and $f(z)$ and $ g(z)$ are polynomials. Suppose that there exists a Lyapunov function $V(z)$ such that $V(z)$ is positive definite in a neighbourhood of the origin and polynomials $p(z), \: q(z)$ satisfying
    \begin{equation*}
        -\frac{\partial V}{\partial z}(f(z) + g(z)u) - (p(z) - q(z)u) \geq 0 \: \forall z,u.
    \end{equation*}
    Then the origin of the state space is a stable equilibrium of the system.
\end{proposition}
\begin{proof}
    We consider the stability of constrained dynamical systems as in \cite{apap1}. Using the same argument as in Proposition \ref{prop:rationalpsatz}, if 
    \begin{equation*}
         -\frac{\partial V}{\partial z}(f(z) + g(z)u),
    \end{equation*}
    is nonnegative on the set 
    \begin{equation*}
        \{ z \in \mathbb{R}^{n_{z}}, \: u \in \mathbb{R}^{n_{u}} \: | \: p(z) - q(z)u = 0, \: q(z) \neq 0 \},
    \end{equation*}
    then it is also nonnegative on the set 
    \begin{equation*}
        \bigg\{ z \in \mathbb{R}^{n_{z}}, \: u \in \mathbb{R}^{n_{u}} \: \bigg| \: \frac{p(z)}{q(z)} - u = 0, \: q(z) \neq 0 \bigg\},
    \end{equation*}
    which defines the controller in the closed-loop system.
\end{proof}

To compute the Lyapunov function and polynomials that define the rational controller in Proposition \ref{prop:rationalcontroller} we can formulate an SOS program.
\begin{proposition} \label{prop:rationalcontrollersos}
    Consider the dynamical system \eqref{eq:systemrational} in Proposition \ref{prop:rationalcontroller}. Suppose there exists polynomial functions $V(z), \: p(z), \: q(z)$, a positive definite function $\rho (z)$ such that
    \begin{align*}
       V(z) - \rho(z) &\in \Sigma[z], \\ 
       -\frac{\partial V}{\partial z}(f(z) + g(z)u) - (p(z) - q(z)u) &\in \Sigma[X], \\ 
        p(z) &\in \mathbb{R}[X],\\ 
        q(z) &\in \mathbb{R}[X],\\
        q(z) &\neq 0,
    \end{align*}
where $X = (z,u)$ is a vector of all of the states. Then the origin of the system is a stable equilibrium.
\end{proposition}

Theorem \ref{prop:rationalcontrollersos} allows us to reconstruct a stabilising controller from a feasibility test using SOS programming. However, the structure of the controller is limited as it is a simple rational function. We therefore extend this approach to a more expressive class of functions through rational NNs.

\subsection{Extension to Rational Neural Network Controllers} \label{sec:extRNNcont}
The technique outlined in the previous section can be used to recover a stabilising controller that is a rational function of the system states. However, we can expand this approach to consider an NN architecture that contains rational functions similar to the one proposed in \eqref{eq:newrationalnn}. We consider a state feedback controller $u(t) = \pi (z(t)) : \mathbb{R}^{n_{z}} \rightarrow \mathbb{R}^{n_{u}}$ as a rational NN such that
\begin{equation} \label{eq:rationalNNcontroller}
\begin{aligned}
    x^{0}(t) &= z(t), \\
    x_{i}^{k+1}(t) &= \frac{p_{i}^{k}(x^{k}(t))}{q_{i}^{k}(x^{k}(t))} , \: \mathrm{for} \: i = 1, \dots, n_{k}, \: k = 0, \dots , \ell - 1, \\
    u_{i}(t) = \pi_{i} (z(t)) &= \frac{p_{i}^{\ell}(x^{\ell}(t))}{q_{i}^{\ell}(x^{\ell}(t))}, \: \mathrm{for}  \: i = 1, \dots , n_{z},
\end{aligned}
\end{equation}
where $p_{i}^{k}(x^{k}(t))$, $q_{i}^{k}(x^{k}(t))$ are the polynomials that form the rational expression associated with the $i^{\text{th}}$ node in the $(k+1)^{\text{th}}$ layer. The number of neurons in the $k^{\text{th}}$ layer is denoted by $n_{k}$. We will drop the time dependence notation throughout the rest of this paper for simplicity.

We can apply Proposition \ref{prop:rationalcontrollersos} and the theory of constrained dynamical systems as in \cite{apap1} due to the well-defined structure of this controller. The following proposition shows how Lyapunov stability over constrained dynamical systems can be used to recover a controller of this form.

\begin{proposition} \label{prop:rationalNNsos}
    Consider the non-linear system
    \begin{align*}
        \dot{z} &= f(z) + g(z)u, \\
        u &= \pi(z),
    \end{align*}
    where $z \in \mathbb{R}^{n_{z}}$ are the system states, $u \in \mathbb{R}^{n_{u}}$ is the controller input and $f(z)$ and $g(z)$ are polynomials. Consider the controller structure $\pi(z)$ defined in \eqref{eq:rationalNNcontroller} and the region given by \eqref{eq:Dz}. Suppose there exist polynomial functions $V(z)$, $p_{i}^{k}(x^{k}) \: \forall i = 1, \dots, n_{k+1}, \: k = 0, \dots , \ell,$ and $q_{i}^{k}(x^{k}) \: \forall i = 1, \dots, n_{k+1}, \: k = 0, \dots , \ell$ satisfying the following conditions
    \begin{align*} 
    \begin{split}
        V(z) - \rho(z) &\in \Sigma[z], \\
        \rho(z) &> 0, \\
        -\frac{\partial V}{\partial z}(z) (f(z) + g(z)u) &- \sum_{k=1}^{n_d}s_{k}(X)d_{k}(z) \dots \\  
        &- \sum_{k=1}^{\ell} \sum_{i=1}^{n_{k}} \bigg( p_{i}^{k}(x^{k})  - q_{i}^{k}(x^{k}) x_{i}^{k+1} \bigg) \dots \\ &- \sum_{i=1}^{n_{u}} \bigg( p_{i}^{\ell}(x^{\ell}) - q_{i}^{\ell}(x^{\ell}) u_{i} \bigg) \in \Sigma[X], \\
        s_{k}(X) &\in \Sigma[X], \: \forall k = 1, \dots , n_{d}, \\
        q_{i}^{k}(x^{k}) &\neq 0, \: \forall \: i = 1, \dots, n_{k}, \: k = 0, \dots , \ell, 
    \end{split}
    \end{align*}
    where $X$ is a vector of all the system and NN states, i.e. $X = (x,u,z)$. Then the equilibrium of the system is stable.
\end{proposition}
The above proposition can generate a feasible SOS program, however the rational NN controller may be difficult to recover. This is because the coefficients in the $q_{i}^{k}(x^{k})$ terms will be set to very small values by the SOS program, due to each term cancelling with the adjacent layers. We therefore propose an alternative rational NN controller structure in the following section to mitigate this issue.

\subsection{Refined Rational Neural Network Controller}
The Lyapunov condition in Proposition \ref{prop:rationalNNsos} may result in numerical issues when solving the SOS program due to the structure of the constraints, making the rational NN controller difficult to recover. To overcome this issue, we enrich the NN structure by considering a state feedback controller $u = \pi (z) : \mathbb{R}^{n_{z}} \rightarrow \mathbb{R}^{n_{u}}$ as a rational NN such that
\begin{equation} \label{eq:modifitedrationalNNcont}
\begin{aligned}
    y^{0} &= z, \\
    x_{i}^{k+1} &= \frac{\sum_{j=1}^{n_{k}}p_{i,j}^{k}(y^{k})y_{j}^{k}}{q_{i}^{k}(z)} , \: \mathrm{for} \: i = 1, \dots, n_{k+1}, \: k = 0, \dots , \ell - 1, \\
    y_{i}^{k+1} &= x_{i}^{k+1} + 1 , \: \mathrm{for} \: i = 1, \dots, n_{k}, \: k = 0, \dots , \ell - 1, \\
    u_{i} = \pi_{i} (z) &=  \frac{\sum_{j=1}^{n_{\ell}}p_{i,j}^{\ell}(y^{\ell})y_{j}^{\ell}}{q_{i}^{\ell}(z)} \Big(\sum_{m = 1}^{n_{z}} z_{m}^{2} \Big), \: \mathrm{for} \: i = 1, \dots , n_{z},
\end{aligned}
\end{equation}
where $p_{i,j}^{k}(x^{k})$, $q_{i}^{k}(z)$ are the  $j^{\text{th}}$ polynomials that form the rational expression associated with the $i^{\text{th}}$ node in the $(k+1)^{\text{th}}$ layer. 

Each $x_{j}^{k}$ node in the NN contains a rational activation function and each $y_{j}^{k}$ node is equal to the $x_{j}^{k}$ term with the addition of a bias term which we set to unity. The $y_{j}^{k}$ term that appears in the numerator of the rational activation function is to ensure that all terms in the $x_{i}^{k+1}$ node are a function of all of the nodes in the $k^{\text{th}}$ layer and to impose more structure on the NN. The denominator $q_{i}^{k}(z)$ is a function of the system states to ensure that all nodes in the network are tied to the system states and not just the nodes in the previous layer. This will ensure that the SOS program does not set the coefficients in the $q_{i}^{k}(z)$ terms to very small values. The final layer contains a multiplier term which is a sum of all of the system states $\sum_{m = 1}^{n_{z}} z_{m}^{2}$ to ensure that the controller input goes to zero at the origin.

We can then adapt Proposition \ref{prop:rationalNNsos} for this modified rational NN structure, to obtain a controller that can be recoverable from the feasibility test.

\begin{proposition} \label{prop:sosmodifiedNNcont}
    Consider the non-linear system
    \begin{align*}
        \dot{z} &= f(z) + g(z)u, \\
        u &= \pi(z),
    \end{align*}
    where $z \in \mathbb{R}^{n_{z}}$ are the system states, $u \in \mathbb{R}^{n_{u}}$ is the controller input and $f(z)$ and $g(z)$ are polynomials. Consider the controller structure $\pi(z)$ defined in \eqref{eq:modifitedrationalNNcont} and the region given by \eqref{eq:Dz}. Suppose there exist polynomial functions $V(z)$, $p_{i,j}^{k}(x^{k}) \: \forall i = 1, \dots, n_{k+1}, \: j = 1, \dots, n_{k}, \: k = 0, \dots , \ell$ and $q_{i}^{k}(z) \: \forall i = 1, \dots, n_{k+1}, \: k = 0, \dots , \ell$ satisfying the following conditions
    \begin{align*} 
    \begin{split}
        V(z) - \rho(z) &\in \Sigma[z], \\
        \rho(z) &> 0, \\
        -\frac{\partial V}{\partial z}(z) (f(z) + g(z)u) &- \sum_{k=1}^{n_d}s_{k}(X)d_{k}(z) \dots \\  
        &- \sum_{k=1}^{\ell} \sum_{i=1}^{n_{k}} \bigg( \bigg( \sum_{j=1}^{n_{k-1}} p_{i,j}^{k}(y^{k})y_{j}^{k} \bigg)  - q_{i}^{k}(z) x_{i}^{k+1} \bigg) \dots \\ 
        &- \sum_{k=1}^{\ell} \sum_{i=1}^{n_{k}} t_{i,k}(X) (y_{i}^{k} - x_{i}^{k} - 1) \dots \\
        &- \sum_{i=1}^{n_{u}} \bigg(  \sum_{j=1}^{n_{\ell}}  \bigg( p_{i,j}^{\ell}(y^{\ell})y_{j}^{\ell} \bigg( \sum_{m=1}^{n_{z}} z_{m}^{2}\bigg) \bigg) - q_{i}^{\ell}(z) u_{i} \bigg)  \in \Sigma[X], \\
        s_{k}(X) &\in \Sigma[X], \: \forall k = 1, \dots , n_{d}, \\
        t_{i,k}(X) &\in \mathbb{R}[X], \: \forall k = 1, \dots, \ell, \: i = 1, \dots, n_{k}, \\
        q_{i}^{k}(x^{k}) &\neq 0, \: \forall \: i = 1, \dots, n_{k}, \: k = 0, \dots , \ell,
    \end{split}
    \end{align*}
    where $X$ is a vector of all the system and NN states, i.e. $X = (u,x,y,z)$. Then the equilibrium of the system is stable.
\end{proposition}
The SOS program in the above proposition is convex in the rational NN parameters and can hence be solved using SOS programming. Note that saturation and any uncertainty and robustness conditions can be incorporated in the same way that is demonstrated in \cite{mnew5}. Proposition \ref{prop:sosmodifiedNNcont} presents a method to obtain a stabilising NN controller for a non-linear polynomial system in a convex way by solving one SOS optimisation problem. As described in Section \ref{sec:rationalNNintro}, other recent approaches such as \cite{mever4,pdon,hyin2,fagu,njun} rely on iterative algorithms or reinforcement learning formulations that are significantly more expensive to compute.

\section{Numerical Examples}
We demonstrate how the approach in Proposition \ref{prop:sosmodifiedNNcont} can be used to recover stabilising rational NN controllers through numerical examples. These examples were run on a four-core Intel Xeon processor @3.50GHz with 16GB of RAM. The SOS programs were implemented using MATLAB and SOSTOOLS to parse the SOS constraints into an SDP, which is solved using MOSEK \cite{mosek}.

\subsection{One Dimensional System} \label{sec:rnneg1D1}
To show that this method is able to obtain a stabilising rational NN controller, we consider a very simple one dimensional linear system of the form
\begin{equation*}
    \dot{z} = z + u.
\end{equation*}
This system is unstable without a feedback controller. To recover a stabilising controller we set the size of the rational NN to be a small two layer network with a single node in each layer. The equations of the controller can be written as
\begin{align*}
    x_{1} &= \frac{\lambda_{1,1} z^{4} + \lambda_{1,2}z^{2}}{\gamma_{1,1}z^2 + \gamma_{1,2}}, \\
    y_{1} &= x_{1} + 1, \\
    u &= \frac{(\lambda_{2,1}y_{1}^{2} + \lambda_{2,2}y_{1}) z}{\gamma_{2,1}z^{2} + \gamma_{2,2}},
\end{align*}
where $\gamma_{1,1} \geq 0$, $\gamma_{1,2} > 0$, $\gamma_{2,1} \geq 0$, $\gamma_{2,2} > 0$. We also add saturation to the controller such that $-10 \leq u \leq 10$ and we enforce the local region of the state space to be $-10 \leq z \leq 10$. The SOS program is able to recover a feasible controller which can stabilise the system to the zero equilibrium. The state trajectory and controller input for this NFL over time are shown in Figure \ref{fig:OneDSymsStates} and \ref{fig:OneDSymsControl} respectively.

\begin{figure}[h] 
    \centering  
    \includegraphics[height=8cm]{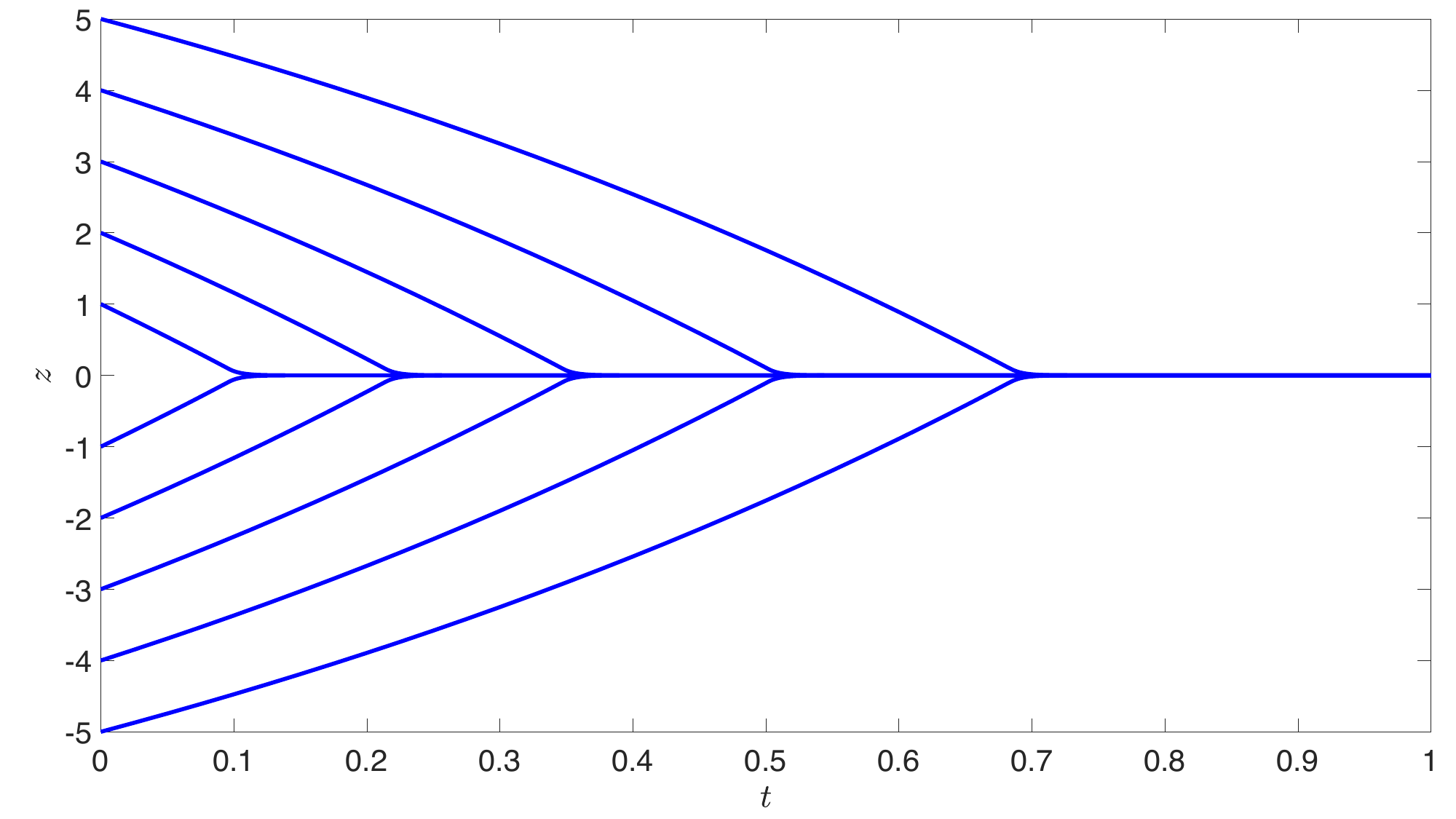}
    \caption{Plot showing the trajectories of the system state in Section \ref{sec:rnneg1D1} over time.} \label{fig:OneDSymsStates}
\end{figure}

\begin{figure}[h] 
    \centering  
    \includegraphics[height=8cm]{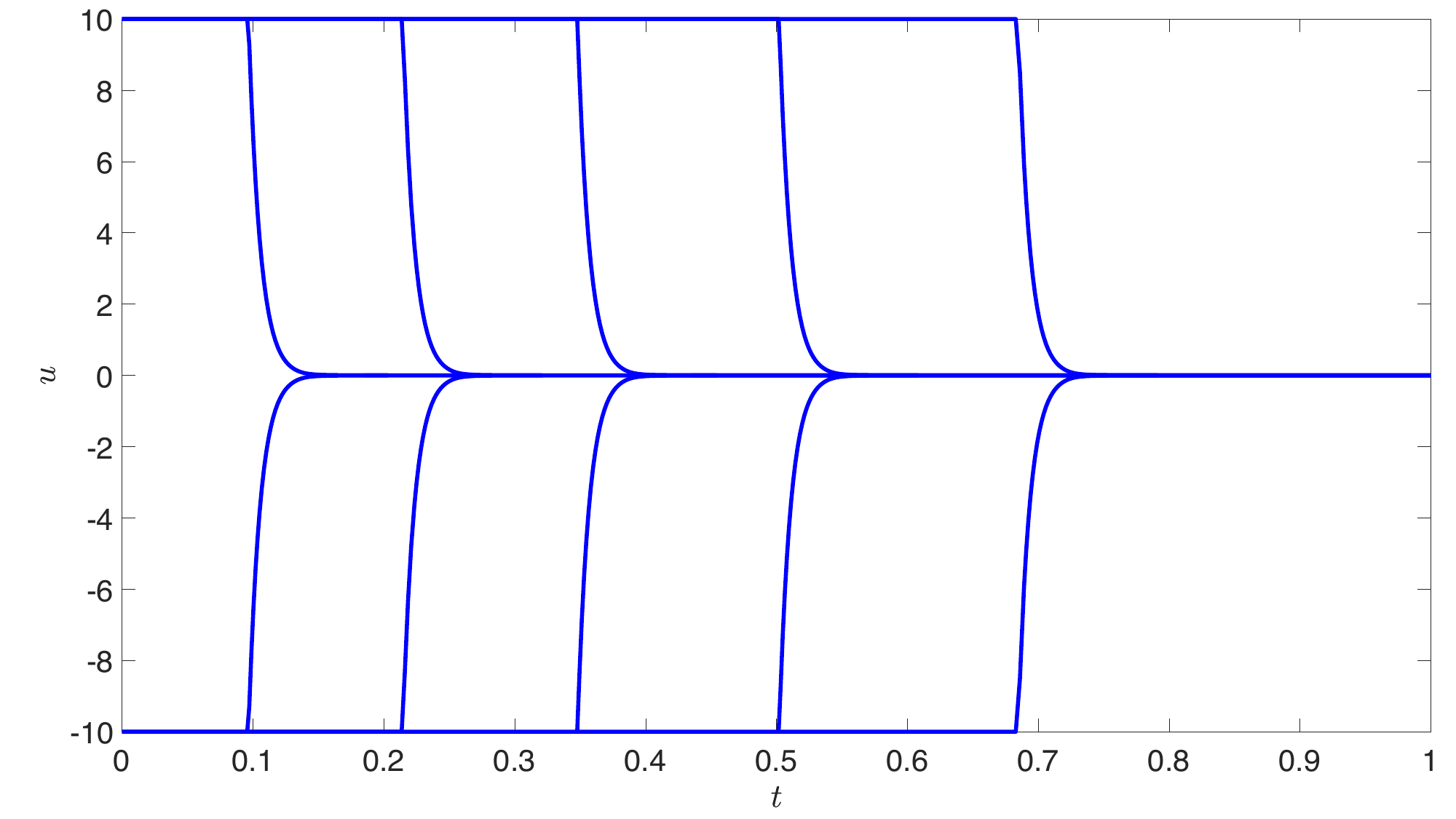}
    \caption{Plot showing the controller input of the system in Section \ref{sec:rnneg1D1} over time.} \label{fig:OneDSymsControl}
\end{figure}

\subsection{Three Dimensional Non-linear System} \label{sec:3Dpolyrationaleg}
We now consider a three dimensional non-linear system system given by
\begin{align*}
    \dot{z}_{1} &= -z_{1} + z_{2} - z_{3}, \\
    \dot{z}_{2} &= -z_{1}(z_{3} + 1) - z_{2}, \\
    \dot{z}_{3} &= -z_{1} + u,
\end{align*}
and attempt to find a stabilising rational NN controller for the system. We set the NN to have two layers and three nodes in each layer, with saturation $-10 \leq u \leq 10$. The system states are defined to operate in the region
\begin{equation*}
    1^2 - z_{1}^{2} - z_{2}^{2} - z_{3}^{2} \geq 0. 
\end{equation*}
Each polynomial in the rational NN is assigned to contain zeroth to fourth order terms. We define a quartic Lyapunov function and second order polynomials for the $s_{k}$ and $t_{i,k}$ terms. The trajectories for this system are shown in Figure \ref{fig:3DSymsStates}, which shows that the controller can successfully stabilise the system. 

\begin{figure}[h] 
    \centering  
    \includegraphics[height=8cm]{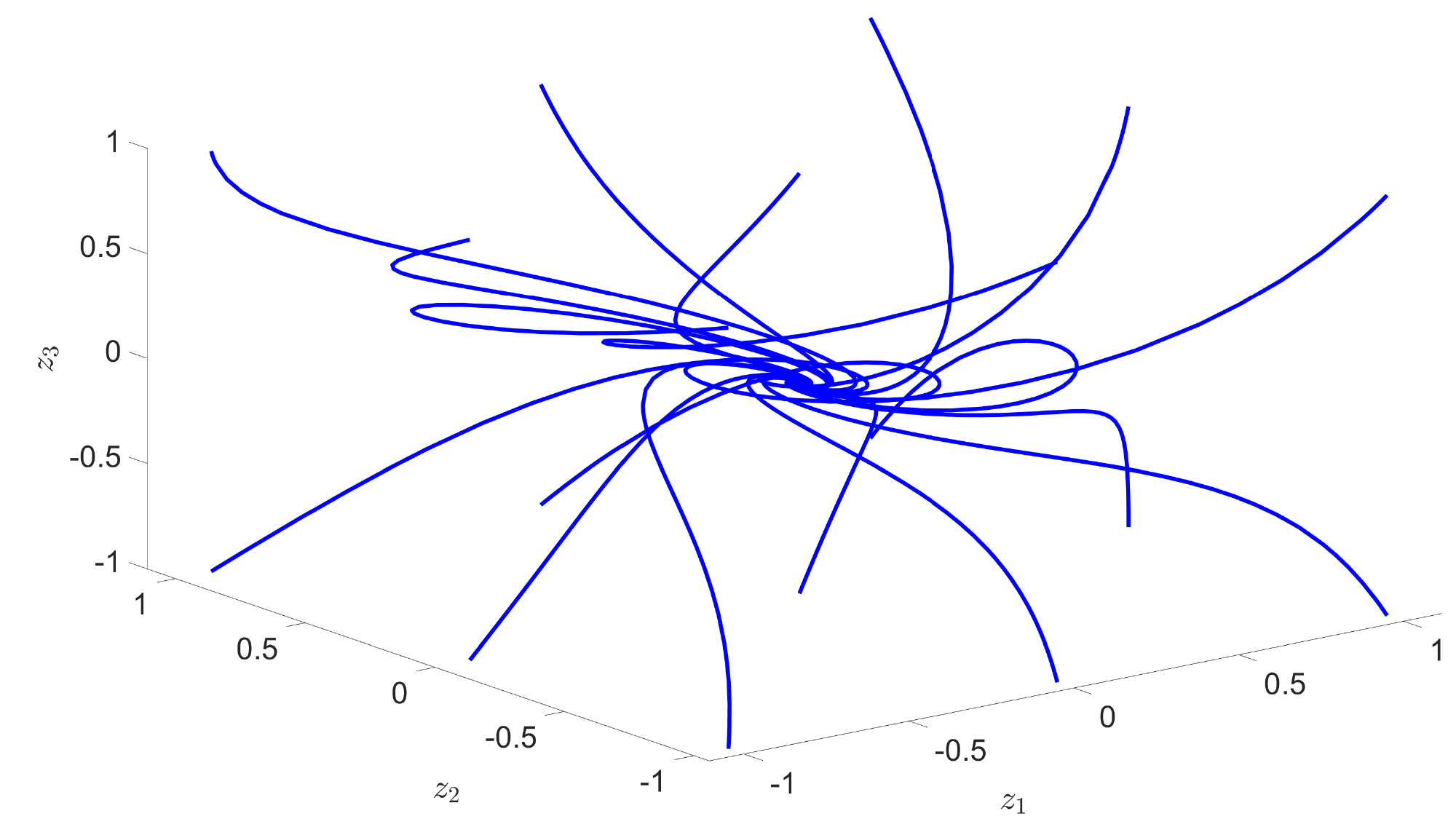}
    \caption{Plot showing the trajectories of the system states for the example in Section \ref{sec:3Dpolyrationaleg}.} \label{fig:3DSymsStates}
\end{figure}

\subsection{Non-linear Inverted Pendulum} \label{sec:IPrationaleg}
We now consider the inverted pendulum proposed in \cite{hyin2} with dynamics given by
\begin{gather*}
    \ddot{\theta}(t) = \frac{mgl\sin{(\theta(t))} - \mu \dot{\theta}(t) + \mathrm{sat}(u(t))}{ml^2}.
\end{gather*}
As shown in \cite{apap1}, we can rewrite the dynamics as a four dimensional polynomial system. We let $z_{1} = \theta$, $z_{2} = \dot{\theta}$ and by making the substitution $z_{3} = \mathrm{sin}(z_{1})$, $z_{4} = \mathrm{cos}(z_{1})$ the system can be written as
\begin{align*}
    \dot{z}_{1} &= z_{2}, \\
    \dot{z}_{2} &= \frac{g}{l}z_{3} - \frac{\mu}{ml^2} z_{2} + \frac{1}{ml^2} u, \\
    \dot{z}_{3} &= z_{2}z_{4}, \\
    \dot{z}_{4} &= -z_{2}z_{3},
\end{align*}
where $m=0.15 \: \mathrm{kg}, \: l=0.5 \: \mathrm{m}, \: \mu=0.5 \: \mathrm{Nmsrad}^{-1}, \: g=9.81 \: \mathrm{ms}^{-2}$ and the controller input is saturated such that $-1 \leq u \leq 1$. The system also requires the equality constraint
\begin{equation*}
    z_{3}^{2} + z_{4}^{2} - 1 = 0,
\end{equation*}
to be enforced. We do not define any region of the state space and instead consider global stability. We include the additional robustness constraints on the length of the pendulum to be $\pm 0.1$ its original length and additive white noise $w$ on the angular velocity such that $|| w ||_{\infty} \leq 0.1$. By making the substitution $\delta = 1/l$ the full dynamical system can be written as
\begin{align*}
    \dot{z}_{1} &= z_{2}, \\
    \dot{z}_{2} &= g \delta z_{3} -  \frac{\mu \delta^2}{m}z_{2} + \frac{\delta^2}{m} u + w, \\
    \dot{z}_{3} &= z_{2}z_{4}, \\
    \dot{z}_{4} &= -z_{2}z_{3}, \\
    0 &= z_{3}^{2} + z_{4}^{2} - 1, \\
    0 &\leq 1^2 - u^2, \\
    0 &\leq 0.1^2 - w^2, \\
    0 &\leq (1/0.4 - \delta)(\delta - 1/0.6).
\end{align*}

The Lyapunov function must be carefully constructed due to the $z_{4}$ state being equal to one at the origin. We therefore define the Lyapunov function to be the sum of two Lyapunov functions, the first of which is defined as
\begin{equation*}
    V_{1}(z_{3},z_{4}) = a_{1}z_{3}^{2} + a_{2}z_{4}^{2} + a_{3}z_{4} + a_{4}
\end{equation*}
and the second is quadratic in $z_{1}$ and $z_{2}$. To ensure that the Lyapunov function is zero at the origin we must enforce
\begin{equation*}
    a_{2} + a_{3} + a_{4} = 0.
\end{equation*}
To ensure that the Lyapunov function is positive definite, we define 
\begin{equation*}
    \rho(z) = \epsilon_{1}z_{1}^{2} + \epsilon_{2}z_{2}^{2} + \epsilon_{3}(1 - z_{4}),
\end{equation*}
where $\epsilon_{1} \geq 0.1$, $\epsilon_{2} \geq 0.1$, $\epsilon_{3} \geq 0.1$.

The size of the rational NN controller is a two layer network with four nodes in each layer. By setting the size of the polynomials in the network to be between zeroth and fourth order, we are able to recover a controller. The trajectories for this NFL are shown in Figure \ref{fig:IPrationalSymsStates}. We can see that the controller initially drives the system states towards a manifold and then moves it towards the equilibrium. Since the control system is discontinuous at the manifold, further analysis is required to show stability of the system.

\begin{figure}[h!] 
    \centering  
    \includegraphics[height=8cm]{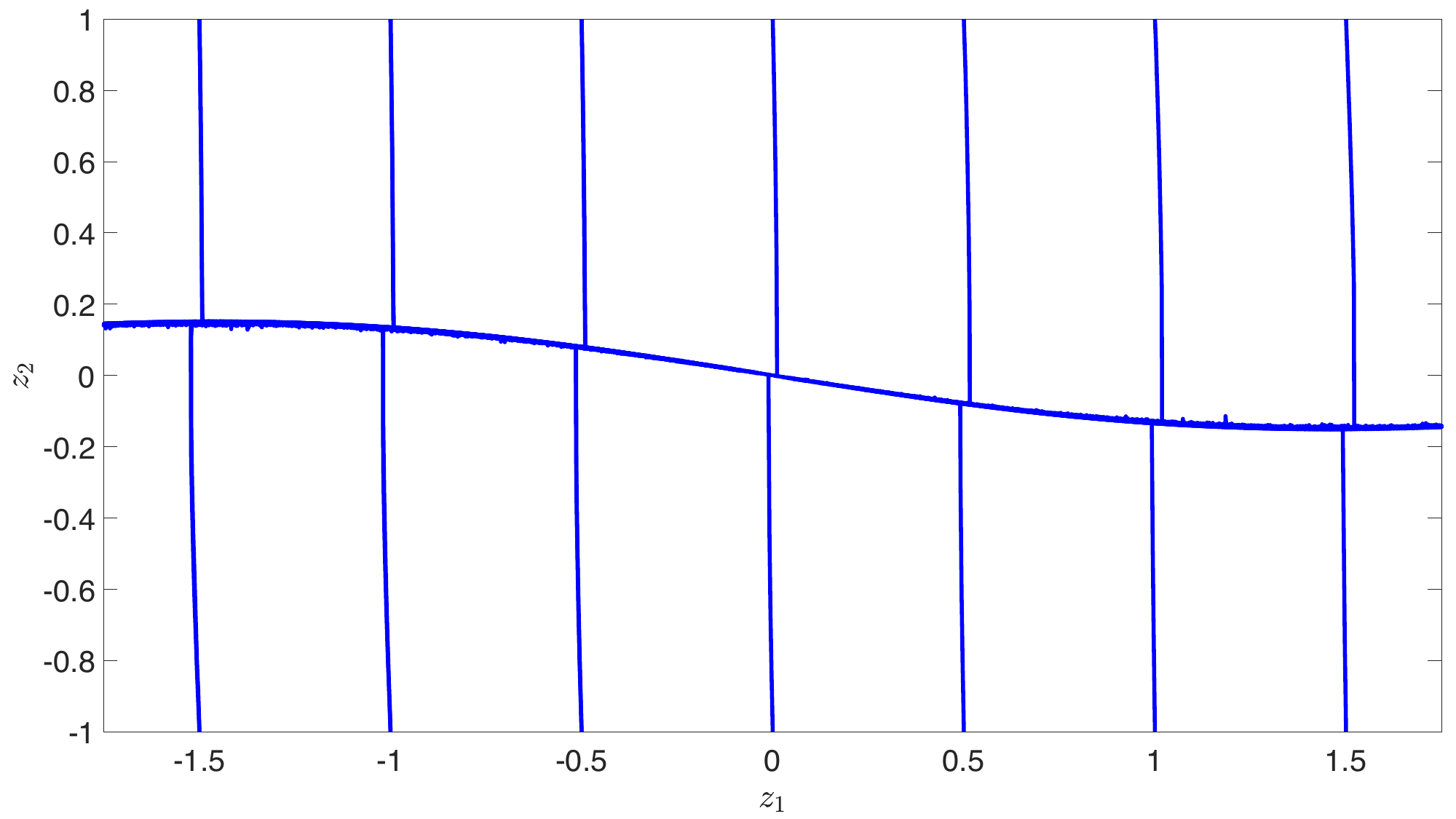}
    \caption{Plot showing the trajectories of the system states for the example in Section \ref{sec:IPrationaleg}.} \label{fig:IPrationalSymsStates}
\end{figure}

\section{Conclusion}
In this paper, we analyse the use of rational NNs in previous application areas. We present novel rational activation functions to approximate the traditional sigmoid and tanh functions and show how they can be used in robustness problems for NFLs. We argue that rational activation functions can be replaced with a general rational NN structure where each layer is convex in the NN's parameters. We then propose a method to recover a stabilising controller from a feasibility test and then extend this approach to rational NNs. This structure is refined to make it more compatible when used in conjunction with SOS programming. Through numerous numerical examples we show how this approach can be used to recover stabilising rational NN controllers for NFLs with non-linear plants with noise and parametric uncertainty.

\section*{Acknowledgements}
This work was supported by EPSRC grants EP/L015897/1 (to M. Newton) and EP/M002454/1 (to A. Papachristodoulou) and the Tony Corner Research Fund. 

\printbibliography

\end{document}